\documentclass[twoside,12pt]{article}

\usepackage[utf8]{inputenc}
\usepackage[T1]{fontenc}
\usepackage{enumitem}
\usepackage{textcomp}
\usepackage{relsize}
\usepackage{setspace}
\usepackage[fleqn]{amsmath}
\usepackage{amssymb}
\usepackage{amsthm}
\usepackage{mathtools}
\usepackage{array}
\usepackage{tabularx}
\usepackage{multicol}
\usepackage{lmodern}
\usepackage[a4paper,left=2.5cm,right=2.5cm,top=2.5cm,bottom=2.5cm,heightrounded]{geometry}
\usepackage{graphicx}
\usepackage[margin=2cm]{caption}
\usepackage[usenames, dvipsnames]{xcolor}
\usepackage{microtype}
\usepackage{hyperref}
\usepackage{cleveref}
\usepackage{url}
\usepackage{tocloft}
\usepackage{titling}
\usepackage{caption}
\hypersetup{pdfstartview=XYZ}
\usepackage{algorithm}
\usepackage{algpseudocode}
\usepackage{float}

\newcommand{\A}{\mathcal{A}}
\renewcommand{\H}{\mathcal{H}}
\newcommand{\I}{\mathcal{I}}
\newcommand{\defeq}{\vcentcolon=}
\renewcommand{\Im}{\textup{Im}}
\newcommand{\FCA}{\textup{LCA}}
\newcommand{\ora}{\overrightarrow}
\renewcommand{\epsilon}{\varepsilon}

\algnewcommand\Initialize{\textbf{initialize}$\,\,$}
\algnewcommand\Define{\textbf{define}$\,\,$}
\algnewcommand\Update{\textbf{update}$\,\,$}

\newtheoremstyle{the}{15pt}{15pt}{\it}{}{\bfseries}{.}{ }{}
\newtheoremstyle{pro}{15pt}{15pt}{\it}{}{\bfseries}{.}{ }{}
\newtheoremstyle{lem}{15pt}{15pt}{\it}{}{\bfseries}{.}{ }{}
\newtheoremstyle{def}{15pt}{15pt}{}{}{\bfseries}{.}{ }{}
\newtheoremstyle{rem}{15pt}{15pt}{}{}{\it}{.}{ }{}
{\theoremstyle{the} \newtheorem{mytheorem}{Theorem}[section]} \crefname{mytheorem}{Theorem}{Theorems}
{\theoremstyle{the} \newtheorem*{mytheorem*}{Theorem}} \crefname{mytheorem*}{Theorem}{Theorems}
{\theoremstyle{lem} \newtheorem{mylemma}[mytheorem]{Lemma}} \crefname{mylemma}{Lemma}{Lemmas}
{\theoremstyle{pro} \newtheorem{myproposition}[mytheorem]{Proposition}} \crefname{myproposition}{Proposition}{Propositions}
{\theoremstyle{the} \newtheorem{mycorollary}[mytheorem]{Corollary}} \crefname{mycorollary}{Corollary}{Corollaries}
{\theoremstyle{def} \newtheorem{mydefinition}[mytheorem]{Definition}} \crefname{mydefinition}{Definition}{Definitions}
{\theoremstyle{def} \newtheorem{mynotation}[mytheorem]{Notation}} \crefname{mynotation}{Notation}{Notations}
{\theoremstyle{rem} \newtheorem*{myremark}{Remark}} \crefname{myremark}{Remark}{Remarks}
{\theoremstyle{rem} \newtheorem*{myexample}{Example}} \crefname{myexample}{Example}{Examples}

\usepackage{authblk}
\newcommand{\nom}[2]{#1 \textsc{#2}}

\title{$(k-2)$-linear connected components in hypergraphs of rank $k$}
\author[1,3]{\nom{Florian}{Galliot}}
\author[1,3]{\nom{Sylvain}{Gravier}}
\author[2,3]{\nom{Isabelle}{Sivignon}}

\affil[1]{Univ. Grenoble Alpes, CNRS, Institut Fourier, 38000 Grenoble, France}
\affil[2]{Univ. Grenoble Alpes, CNRS, Grenoble INP, GIPSA-lab, 38000 Grenoble, France}
\affil[3]{Univ. Grenoble Alpes, Maths à Modeler, 38000 Grenoble, France}

\date{}

\begin{document}
\maketitle

\begin{abstract}
  We define a \textit{$q$-linear path} in a hypergraph $\mathcal{H}$ as a sequence $(e_1,\ldots,e_L)$ of edges of $\mathcal{H}$ such that $|e_i \cap e_{i+1}|\in  [\![1,q]\!]$ and $e_i \cap e_j = \varnothing$ if $|i-j|>1$. In this paper, we study the connected components associated to these paths when $q=k-2$ where $k$ is the rank of $\mathcal{H}$. If $k=3$ then $q=1$ which coincides with the well-known notion of \textit{linear path} or \textit{loose path}. We describe the structure of the connected components, using an algorithmic proof which shows that the connected components can be computed in polynomial time. We then mention two consequences of our algorithmic result. The first one is that deciding the winner of the Maker-Breaker game on a hypergraph of rank 3 can be done in polynomial time. The second one is that tractable cases for the NP-complete problem of "Paths Avoiding Forbidden Pairs" in a graph can be deduced from the recognition of a special type of line graph of a hypergraph.
\end{abstract}

\section*{Introduction} 

\hphantom{\indent}There are many possible definitions for a path between two vertices in a hypergraph. Each one has its own associated connectivity problem, consisting in the algorithmic computation of the connected components and the potential study of their structure. Possible fields where such problems apply include system security \cite{GPR14} on undirected hypergraphs as well as propositional logic \cite{GLP93}, system transfer protocols \cite{TT09} or computational tropical geometry \cite{All14} on directed hypergraphs.
\medskip
\\ \indent In an undirected hypergraph, a \textit{linear path} (or \textit{loose path}) is a sequence of edges such that any two consecutive edges intersect on exactly one vertex and any two non-consecutive edges do not intersect. Our main motivation is the connectivity problem associated with linear paths in 3-uniform hypergraphs. The existence of such paths is the subject of numerous extremal results \cite{OS14} \cite{Jac15} \cite{JPR16} \cite{WP21}. For instance, \cite{JPR16} determines the Tur\'an number of the 3-uniform linear path of length 3, so that a 3-uniform hypergraph on $n \geq 8$ vertices with at least $\binom{n-1}{2}$ edges necessarily contains a 3-uniform linear path of length 3. Such results are proven using  counting methods. The study of linear structures in potentially sparser hypergraphs, however, requires tools of a qualitative nature. It then seems reasonable to start by studying the linear connected components. In order to describe their structure, we develop methods that actually generalize to hypergraphs of rank $k \geq 4$ when replacing linearity with a notion of \textit{$(k-2)$-linearity}.
\medskip
\\ \indent We thus introduce the general concept of $q$-linear path, where any two consecutive edges intersect on between 1 and $q$ vertices (and non-consecutive edges do not intersect). Extremal results also exist on paths with similar restrictions on the size of the intersections, for example paths where any two consecutive edges must intersect on exactly $q$ vertices \cite{Tom12} \cite{DLM17} with emphasis on the linear case $q=1$ \cite{FJS14} \cite{GLS20}. Throughout this article, let $\H$ be a hypergraph of rank $k$: as for any hypergraph, we denote its vertex set by $V(\H)$ and its edge set by $E(\H)$. Define the $q$-linear connected component of $x^* \in V(\H)$ as the set $LCC^{\,q}_{\H}(x^*)$ of all vertices $x$ such that there exists a $q$-linear path between $x^*$ and $x$ in $\H$. We will see that $q$-linear paths do not define a transitive relation, so that the $q$-linear connected components of $\H$ do not form a partition of $V(\H)$, unlike most other connectivity problems. This paper is a study of the $q$-linear connected components of $\H$ in the case $q=k-2$, meaning that we only prohibit \textit{tight} intersections of size $k-1$. Linear paths in 3-uniform hypergraphs correspond to the case $k=3$ i.e. $q=1$. Our first main result describes the structure of the subhypergraph $\H[LCC^{\,k-2}_{\H}(x^*)]$ induced by a $(k-2)$-linear connected component.
\medskip
\\ \indent The proof of the structural result is algorithmic and provides us with a way to compute the $(k-2)$-linear connected components in polynomial time. More precisely, our second main result is an algorithm that computes $LCC^{\,k-2}_{\H}(x^*)$ in $O(m^2k)$ time where $m=|E(\H)|$, which remains polynomial even if $k$ is part of the input. This result has consequences on two algorithmic problems that have long existed in the literature.
\\ \indent The first one is the problem of deciding the winner of the \textit{Maker-Breaker} positional game. Two players, Maker and Breaker, take turns picking vertices of a hypergraph $\H$: Maker wins if he owns all the vertices of some edge of $\H$, and Breaker wins if he prevents this from happening. The problem of deciding the winner of the game with optimal play is trivially tractable for hypergraphs of rank 2, and is known to be PSPACE-complete for 6-uniform hypergraphs \cite{RW21}. In a separate paper \cite{GGS22}, we show tractability for hypergraphs of rank 3, by reducing to the linear path existence problem in 3-uniform hypergraphs and using the polynomial-time algorithm provided by the present paper. This validates a conjecture by Rahman and Watson \cite{RW20}.
\\ \indent The second one is the "Paths Avoiding Forbidden Pairs" problem (known as $\textsc{PAFP}$) which, given two vertices $x,y$ in a graph $G$ with blue and red edges, asks whether there exists a blue induced path between $x$ and $y$ in $G$. Indeed, consider a bicolored version of the line graph of a hypergraph, where a blue (resp. red) edge indicates an intersection of size between 1 and $k-2$ (resp. of size $k-1$): if $G$ is the bicolored line graph of some $k$-uniform hypergraph $\H$, then there exists a blue induced path between two vertices of $G$ if and only if there exists a $(k-2)$-linear path in $\H$ between the corresponding (hyper)edges. Since our connectivity problem is solvable in polynomial time, the study of the bicolored line graph recognition problem has the potential to unearth new tractable cases for $\textsc{PAFP}$, which is known to be NP-complete in general \cite{GMO76}.
\medskip
\\ \indent After some basic definitions given in \Cref{Section1}, including the introduction of $q$-linear paths, \Cref{Section2} presents structures that are specific to the case $q=k-2$ as well as some of their properties. It is then shown algorithmically in \Cref{Section3} that these structures describe the $(k-2)$-linear connected components, which can be computed in polynomial time: these are our two main results. Finally, \Cref{Section4} addresses the links that our algorithmic problem has with the Maker-Breaker game and the $\textsc{PAFP}$ problem. We end by formulating some open problems that arise from our study.

\section{$q$-linear paths}\label{Section1}

\subsection{Sequences of edges}

\begin{mydefinition}
	A \textit{sequence of edges of $\H$} is some $\ora{P}=(e_1,\ldots,e_L)$ where $e_i \in E(\H)$ for all $1 \leq i \leq L$. The case $L=0$ is authorized: we may then denote $\ora{P}=()$.
\end{mydefinition}

\begin{mynotation}
	Let $\ora{P}=(e_1,\ldots,e_L)$ be a sequence of edges of $\H$.
	\begin{itemize}[noitemsep,nolistsep]
		\item We define $V(\ora{P}) \defeq e_1 \cup \ldots \cup e_L \subseteq V(\H)$ and $E(\ora{P}) \defeq \{e_1,\ldots,e_L\} \subseteq E(\H)$.
		\item Let $\ora{Q}=(e'_1,\ldots,e'_M)$ be another sequence of edges of $\H$. We denote by $\ora{P} \oplus \ora{Q}$ the concatenation of $\ora{P}$ and $\ora{Q}$, that is $\ora{P} \oplus \ora{Q} \defeq (e_1,\ldots,e_L,e'_1,\ldots,e'_M)$.
	\end{itemize}
\end{mynotation}

\subsection{Description of the problem}

\begin{mydefinition}
	 A \textit{path} in $\H$ is a sequence $\ora{P}=(e_1,\ldots,e_L)$ of edges of $\H$ such that one can write $V(\ora{P})=\{x_1,\ldots,x_N\}$ and
$e_i=\{x_{s_i},x_{s_i+1},\ldots,x_{f_i}\}$ with $s_i<s_{i+1}\leq f_i<f_{i+1}$ for all $1 \leq i \leq L-1$. Note that $e_i \cap e_{i+1} \neq \varnothing$ for all $1 \leq i \leq L-1$. The path is deemed \textit{simple} if $e_i \cap e_j = \varnothing$ for all $1 \leq i,j \leq L$ such that $|i-j|>1$. See \Cref{SimplePath}.
\end{mydefinition}

\begin{figure}[htbp]
	\centering
	\includegraphics[scale=.5]{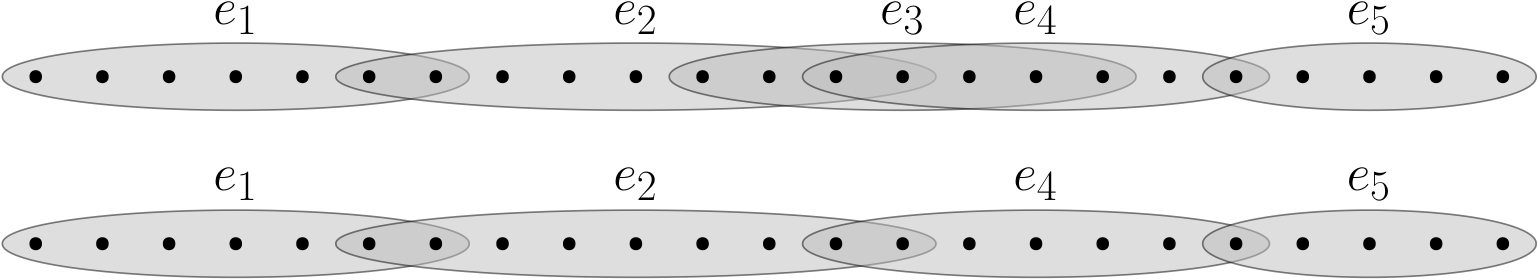}
	\caption{The top path is not simple because $e_2 \cap e_4 \neq \varnothing$. Removing $e_3$ yields a simple path (bottom).}\label{SimplePath}
\end{figure}

\indent We study paths with the additional \textit{$q$-linearity property} that $|e_i \cap e_{i+1}| \leq q$ for some fixed integer $q$. Since we are only interested in existence questions, we can focus on simple such paths: indeed, from any path it is possible to extract a simple path by removing some edges if necessary, and this obviously preserves the $q$-linearity property. An equivalent definition is the following:

\begin{mydefinition}
	Let $q \geq 1$. A \textit{$q$-linear path} in $\H$ is a sequence $\ora{P}=(e_1,\ldots,e_L)$ of edges of $\H$ such that for all $1 \leq i < j \leq L$: $|e_i \cap e_j| \,\,\begin{cases} \,\,\in [\![1,q]\!] & \text{if } j=i+1. \\  \,\,=0 & \text{otherwise.} \end{cases}$.
\end{mydefinition}

\begin{mydefinition}\label[mydefinition]{def-q-linear}
	Let $q \geq 1$ be an integer and let $X,Y \subseteq V(\H)$ be nonempty such that $|X \cap Y| \leq q$. A \textit{$q$-linear path from $X$ to $Y$} in $\H$ is a $q$-linear path $\ora{P}=(e_1,\ldots,e_L)$ in $\H$ such that:
	\begin{itemize}[noitemsep,nolistsep]
		\item If $X \cap Y \neq \varnothing$, then $L=0$.
		\item If $X \cap Y = \varnothing$, then $L \geq 1$ and:
			\begin{enumerate}[noitemsep,nolistsep,label=(\roman*)]
				\item $X \cap e_1 \neq \varnothing$, and if $L \geq 2$ then $X \cap e_i = \varnothing$ for all $2 \leq i \leq L$.
				\item $Y \cap e_L \neq \varnothing$, and if $L \geq 2$ then $Y \cap e_i = \varnothing$ for all $1 \leq i \leq L-1$.
			\end{enumerate}
	\end{itemize}
	Whenever $X=\{x\}$, we may use the abuse of notation $X=x$ (same for $Y$). See \Cref{q-linear_path}.
\end{mydefinition}

\begin{figure}[htbp]
	\centering
	\includegraphics[scale=.5]{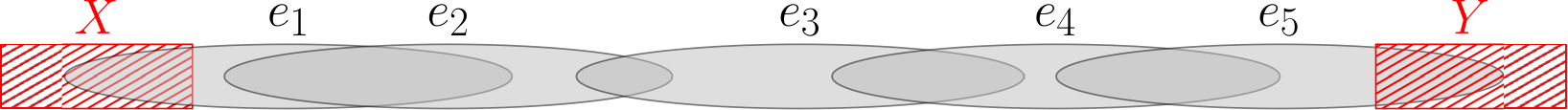}
	\caption{Schematic representation of a $q$-linear path from $X$ to $Y$.}\label{q-linear_path}
\end{figure}

\begin{mylemma}\label[mylemma]{Lemma-subpath}
	Let $\ora{P}=(e_1,\ldots,e_L)$ be a $q$-linear path in $\H$ such that $L \geq 1$. Let $X,Y \subseteq V(\H)$ be disjoint such that $X \cap e_1 \neq \varnothing$ and $Y \cap e_L \neq \varnothing$. Then $\ora{P}$ contains a $q$-linear path $\ora{Q}$ from $X$ to $Y$ in $\H$. More precisely: $\ora{Q}=(e_r,\ldots,e_s)$ where  $s \defeq \inf\{1 \leq i \leq L \,\,\text{such that $e_i \cap Y \neq \varnothing$} \}$ and $r \defeq \sup\{1 \leq i \leq s \,\,\text{such that $e_i \cap X \neq \varnothing$} \}$.
\end{mylemma}

\begin{proof}
	This is clear by minimality (resp. maximality) of $s$ (resp. $r$).
\end{proof}

\begin{mydefinition}
	Let $x \in V(\H)$. The \textit{$q$-linear connected component of $x$} in $\H$ is defined as:
	$$ LCC^{\,q}_{\H}(x) \defeq \{y \in V(\H) \,\,\text{such that there exists a $q$-linear path from $x$ to $y$ in $\H$} \}. $$
\end{mydefinition}

\indent It is important to note that $q$-linear paths do not define a transitive relation, so that the $q$-linear connected components of a hypergraph do not necessarily form a partition of its vertex set. Indeed, the union of a $q$-linear path from $x$ to $y$ and a $q$-linear path from $y$ to $z$ does not necessarily contain a $q$-linear path from $x$ to $z$. An illustration in the case $q=1$ is provided in \Cref{NonTransitive} (this graphical representation of 3-uniform hypergraphs will be used throughout, with each edge pictured as a "claw" joining its three vertices). Therefore, the problem consisting in computing the $q$-linear connected component of a given vertex is nontrivial.

\begin{figure}[htbp]
	\centering
	\includegraphics[scale=.6]{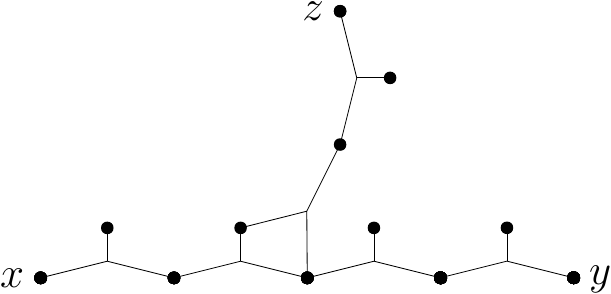}
	\caption{There is no $1$-linear path from $x$ to $z$.}\label{NonTransitive}
\end{figure}

\indent This problem reduces polynomially to the case where $\H$ is uniform. Indeed, if $\H$ is of rank $k$ then let $\H_0$ be the $k$-uniform hypergraph obtained from $\H$ by adding $k-|e|$ new vertices to each edge $e$: it is easy to see that there exists a $q$-linear path from $x$ to $y$ in $\H$ if and only if there exists one in $\H_0$. We thus introduce the following decision problem: \\

\begin{tabularx}{0.95\textwidth}{|l @{} l @{} X|}
	\hline
	\multicolumn{3}{|l|}{$\,\,\textsc{HypConnectivity}_{k,q}$} \\ \hline
	Input $\,$ & : & $\,$ a $k$-uniform hypergraph $\H$ and two distinct vertices $x,y$ of $\H$. \\
	Output $\,$ & : & $\,$ YES if and only if there exists a $q$-linear path from $x$ to $y$ in $\H$. \\ \hline
\end{tabularx} \\

\indent The case $q=k-1$ corresponds to standard (i.e. non-constrained) connectivity in hypergraphs, which is tractable via a simple DFS/BFS-type search. We now address the case $q=k-2$.

\section{$(k-2)$-linear paths in $k$-uniform hypergraphs}\label{Section2}

\hphantom{\indent}In this section, we suppose $\H$ is $k$-uniform with $k \geq 3$.

\subsection{Extendable paths and islands}

\subsubsection{Principle}

\hphantom{\indent}Let $x^* \in V(\H)$ be the vertex whose $(k-2)$-linear connected component we wish to compute. The idea is to design an algorithm that searches through $E(\H)$ and accepts edges under some guarantee that all their vertices are in $LCC^{\,k-2}_{\H}(x^*)$.
\\ \indent Consider the situation in the middle of the execution of the algorithm. Some edges have already been accepted, forming a subhypergraph $\I_1$ of $\H$ containing $x^*$ such that: for all $x \in V(\I_1)$, there exists a $(k-2)$-linear path from $x^*$ to $x$ in $\I_1$. Now, the algorithm encounters some edge $e$ intersecting both $V(\I_1)$ and $V(\H)\setminus V(\I_1)$, and needs to decide whether or not $e$ should be accepted right away: let $x \in e \setminus V(\I_1)$, can we find a $(k-2)$-linear path from $x^*$ to $x$ made of edges in $E(\I_1) \cup \{e\}$?
\\ \indent The only way would be to use a $(k-2)$-linear path $\ora{P}=(e_1,\ldots,e_L)$ from $x^*$ to $X \defeq e \cap V(\I_1)$ in $\I_1$ (\Cref{Lemma-subpath} ensures there exists one), and prolong it with the edge $e$ to reach $e \setminus V(\I_1)$. However, though $\ora{P} \oplus (e)=(e_1,\ldots,e_L,e)$ is obviously $(k-2)$-linear if $|X| \leq k-2$, it might not be if $|X|=k-1$: indeed, in that case, if $X \subset e_L$ then $|e_L \cap e|=k-1$. On this account, if $|X|=k-1$ then we need $\ora{P}$ to not just be any $(k-2)$-linear path from $x^*$ to $X$ but to be one that satisfies $X \not\subset e_L$: such a path will be deemed \textit{$(x^*,X)$-extendable}, because it can be prolonged by an edge that contains $X$ while preserving the $(k-2)$-linearity. An illustration is given in \Cref{Principle}.
\\ \indent So, what property must $\I_1$ have if we want to be able to accept any edge intersecting both $V(\I_1)$ and $V(\H)\setminus V(\I_1)$? As we have just seen, the existence of a $(k-2)$-linear path from $x^*$ to $x$ in $\I_1$ for all $x \in V(\I_1)$ is not sufficient. Additionally to this, we would need the existence of an $(x^*,X)$-extendable path in $\I_1$ for all $X \subset V(\I_1)$ of size $k-1$. If $\I_1$ satisfies these two properties, we will say $\I_1$ is an \textit{island with entry $\{x^*\}$}.

\begin{figure}[htbp]
	\centering
		\includegraphics[scale=.6]{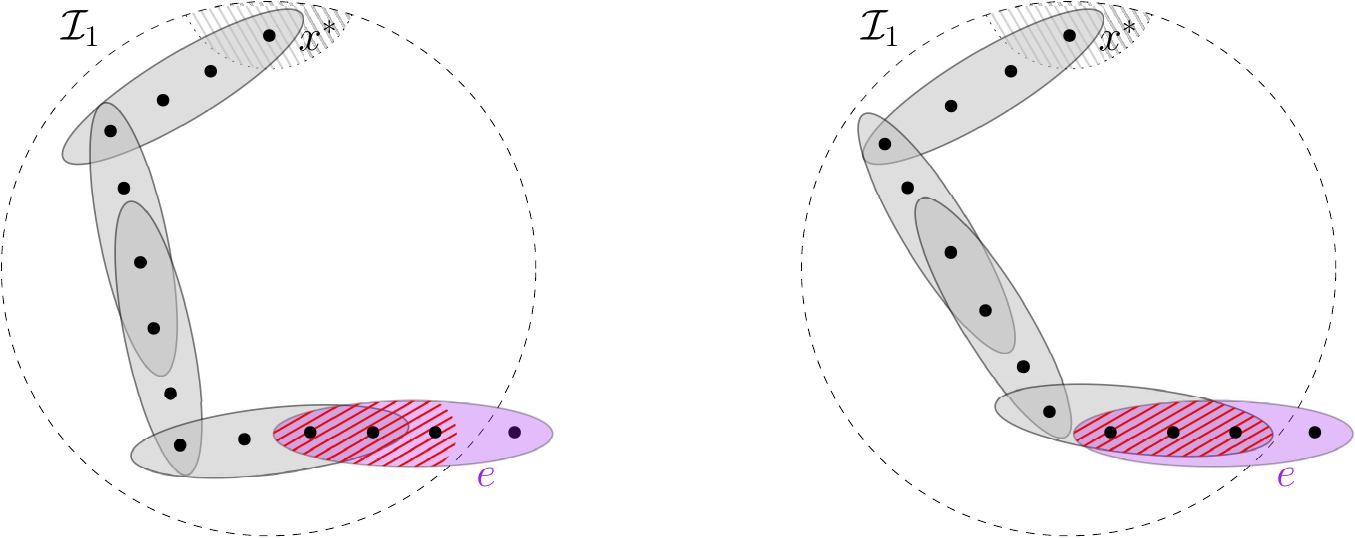}
	\caption{Here $k=4$ and $|X|=3$ (the red hatched area is $X$). The grey path from $x^*$ to $X$ on the left is $(x^*,X)$-extendable, but the one on the right is not because its final edge contains $X$ entirely.}\label{Principle}
\end{figure}

\indent However, the accepted edges might not always form an island. Suppose $\I_1$ is an island and we next discover an edge $e_0$ such that $|e_0 \cap V(\I_1)|=1$ (so we accept $e_0$) i.e. $e_0$ is of the form $e_0=\{x_1\} \cup \epsilon$ where $e_0 \cap V(\I_1)=\{x_1\}$ and $|\epsilon|=k-1$. Then the accepted edges do not form an island anymore: the only known $(k-2)$-linear paths from $x^*$ to $\epsilon$ use $e_0$ so they contain $\epsilon$ entirely, meaning they are not $(x^*,\epsilon)$-extendable. Suppose the next few accepted edges form a subhypergraph $\I_2$ that contains $\epsilon$ but is disjoint from $\I_1$, such that for all $x \in V(\I_2)$ there exists a $(k-2)$-linear path from $\epsilon$ to $x$ in $\I_2$. The algorithm now encounters some edge $e$ whose known vertices are in $\I_2$ (see \Cref{Principle2}): should we accept $e$? Let $X \defeq e \cap V(\I_2)$ and $y \in e\setminus X$. The only way to reach $y$ from $x^*$ is via $\ora{R} \defeq \ora{P} \oplus (e_0) \oplus \ora{Q} \oplus (e)$ where $\ora{P}$ is a $(k-2)$-linear path from $x^*$ to $x_1$ in $\I_1$ and $\ora{Q}$ is a $(k-2)$-linear path from $\epsilon$ to $X$ in $\I_2$. We know such a $\ora{P}$ exists, however there are conditions on $\ora{Q}=(e_1,\ldots,e_L)$ for $\ora{R}$ to be $(k-2)$-linear:
\begin{itemize}[noitemsep,nolistsep]
	\item As before, if $|X|=k-1$ then we need $X \not\subset e_L$.
	\item Since $\epsilon \subset e_0$, we also need $\epsilon \not\subset e_1$.
\end{itemize}
\hphantom{\indent}Such a path $\ora{Q}$ will be deemed \textit{$(\epsilon,X)$-extendable} (this time, there are conditions at both ends of the path). In conclusion, to be able to accept any such $e$, we would need the existence of an $(\epsilon,X)$-extendable path in $\I_2$ for all $X \subset V(\I_2)$ of size at most $k-1$. If $\I_2$ satisfies these two properties, we will say $\I_2$ is an \textit{island with entry $\epsilon$}.

\begin{figure}[htbp]
	\centering
	\includegraphics[scale=.6]{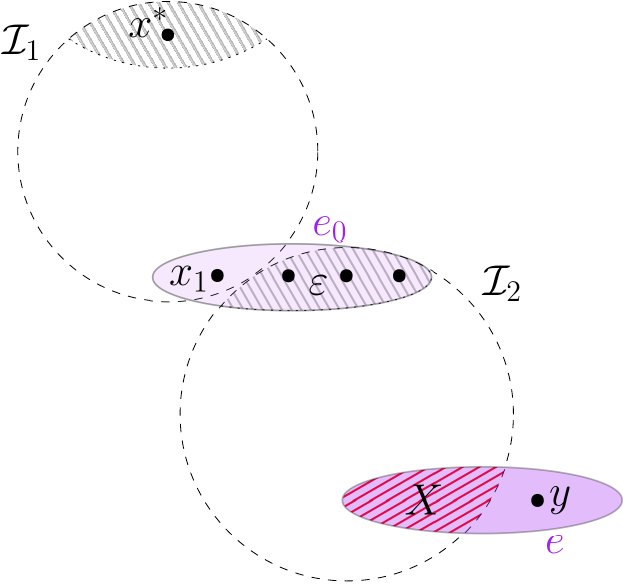}
	\caption{Here $k=4$ so $|\epsilon|=3$.}\label{Principle2}
\end{figure}

\indent We see the premises of the \textit{archipelago} structure of $\H[LCC^{\,k-2}_{\H}(x^*)]$, which we are going to establish.

\subsubsection{Definitions}

\hphantom{\indent}We now give the formal definitions that we are going to use.

\begin{mydefinition}
	Let $X,Y \subseteq V(\H)$ such that $1 \leq |X|,|Y| \leq k-1$ and $|X \cap Y| \leq k-2$. An \textit{$(X,Y)$-extendable path} in $\H$ is a $(k-2)$-linear path $\ora{P}=(e_1,\ldots,e_L)$ from $X$ to $Y$ in $\H$ with the additional property if $L \geq 1$ that $|e_1 \cap X| \leq k-2$ and $|e_L \cap Y| \leq k-2$.
\end{mydefinition}

\begin{figure}[htbp]
	\centering
	\includegraphics[scale=.6]{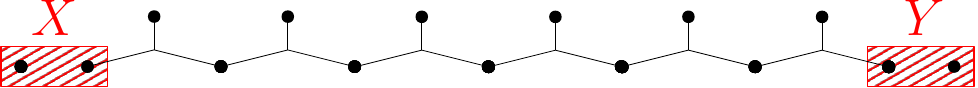}
	\caption{An $(X,Y)$-extendable path in the case $k=3$: the path contains exactly one vertex of $X$ and one vertex of $Y$.}\label{ExtendablePath}
\end{figure}

\indent Note that the condition on $X$ is empty if $|X|\leq k-2$: it is only when $|X|=k-1$ that we need to make sure that prolonging $\ora{P}$ with an edge containing $X$ maintains the $(k-2)$-linearity (same for $Y$). Therefore, if $|X|,|Y| \leq k-2$, then an $(X,Y)$-extendable path is simply a $(k-2)$-linear path from $X$ to $Y$. It is also important to keep in mind that the definition is dependent on $X$ and $Y$: we do not define an "extendable path", we define an "$(X,Y)$-extendable path".

\begin{mydefinition}
	Let $\I$ be a subhypergraph of $\H$ and $\epsilon \subset V(\I)$ such that $1 \leq |\epsilon| \leq k-1$. We say $\I$ is an \textit{island with entry $\epsilon$} if, for all $X \subset V(\I)$ satisfying $1 \leq |X| \leq k-1$ (and $X \neq \epsilon$ if $|\epsilon|=k-1$), there exists an $(\epsilon,X)$-extendable path in $\I$.
\end{mydefinition}

\begin{myexample}
	The \textit{empty island with entry $\epsilon \subset V(\H)$}, where $1 \leq |\epsilon| \leq k-1$, is the island $\I$ with entry $\epsilon$ defined by $V(\I)=\epsilon$ and $E(\I)=\varnothing$. It is an island because, for all $X \subset V(\I)$ satisfying $1 \leq |X| \leq k-1$ (and $X \neq \epsilon$ if $|\epsilon|=k-1$), $\ora{P}=()$ is an $(\epsilon,X)$-extendable path in $\I$. This example is illustrated at the far left of \Cref{Islands}. 
\end{myexample}

\begin{figure}[htbp]
	\centering
	\includegraphics[scale=.6]{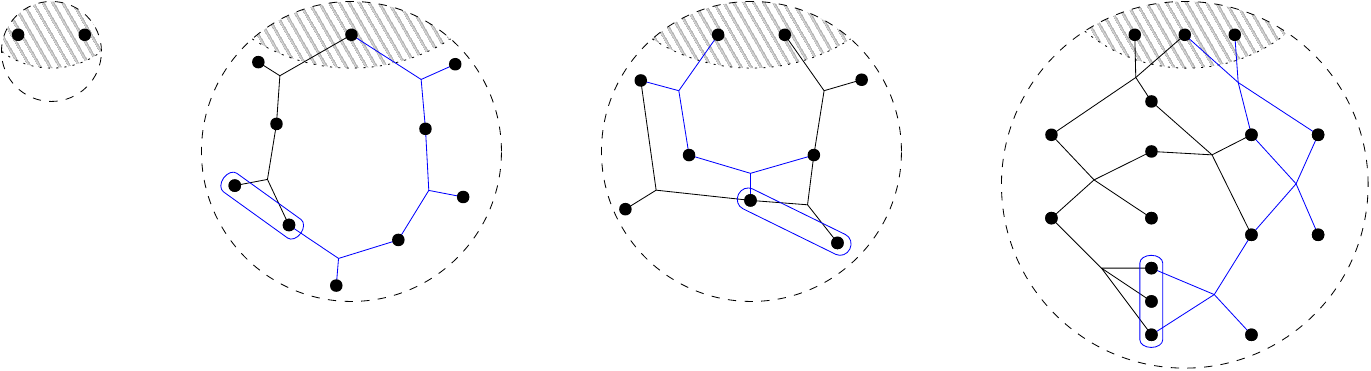}
	\caption{Some islands for $k=3$, except the far right one where $k=4$ (with the same "claw" representation for edges). The grey hatched area will always represent the entry. For three of them, we show an $(\epsilon,X)$-extendable path (in blue) for some $X$ of size $k-1$ (circled in blue).}\label{Islands}
\end{figure}

\subsubsection{Extension lemmas}

\hphantom{\indent}The notion of $(X,Y)$-extendable path has been introduced to prolong and compose $(k-2)$-linear paths. In that direction, we now prove two useful lemmas which are illustrated in \Cref{ExtensionLemma1,ExtensionLemma3}.

\begin{mylemma}\label[mylemma]{Lemma1}
	Let $A,B \subseteq V(\H)$ such that $1 \leq |A|,|B| \leq k-1$ and $|A \cap B| \leq k-2$, and let $\ora{P}$ be an $(A,B)$-extendable path.
	\begin{itemize}[noitemsep,nolistsep]
		\item If $B' \supset B$ is such that $|B'| \leq k-1$ and $B' \cap (A \cup V(\ora{P}) \cup B) = B$, then $\ora{P}$ is an $(A,B')$-extendable path.
		\item If $A' \supset A$ is such that $|A'| \leq k-1$ and $A' \cap (A \cup V(\ora{P}) \cup B) = A$, then $\ora{P}$ is an $(A',B)$-extendable path.
	\end{itemize}
\end{mylemma}

\begin{proof}
	By symmetry, we only need to prove the first assertion. First notice that $A \cap B = A \cap B'$, so that $|A \cap B'| \leq k-2$ as required in \Cref{def-q-linear}.
	\begin{itemize}[noitemsep,nolistsep]
		\item If $A \cap B' \neq \varnothing$ then $A \cap B \neq \varnothing$, hence $\ora{P}=()$ which is an $(A,B')$-extendable path.
		\item If $A \cap B' = \varnothing$ then $A \cap B = \varnothing$, so we can write $\ora{P}=(e_1,\ldots,e_L)$ where $L \geq 1$. We already know $\ora{P}$ is $(k-2)$-linear, moreover the assumption on $B'$ ensures that $\ora{P}$ is from $A$ to $B'$. Finally, since $\ora{P}$ is $(A,B)$-extendable and $e_L \cap B'=e_L \cap B$, we have $|e_1 \cap A| \leq k-2$ and $|e_L \cap B'|=|e_L \cap B| \leq k-2$, therefore $\ora{P}$ is $(A,B')$-extendable. \qedhere
	\end{itemize}
\end{proof}

\begin{figure}[htbp]
	\centering
	\includegraphics[scale=.49]{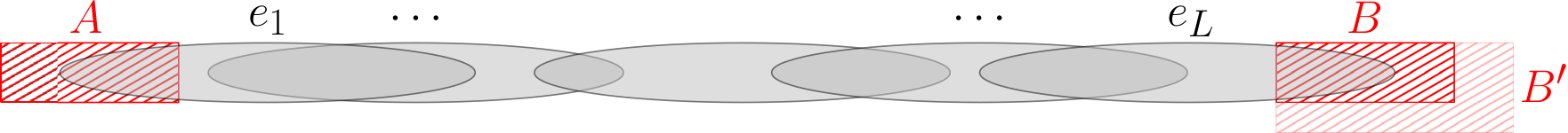}
	\caption{Illustration of \Cref{Lemma1}.}\label{ExtensionLemma1}
\end{figure}
	
\begin{mylemma}\label[mylemma]{Lemma3}
	Let $A,B \subseteq V(\H)$ such that $1 \leq |A|,|B| \leq k-1$ and $|A \cap B| \leq k-2$, and let $\ora{P}$ be an $(A,B)$-extendable path. Let $C,D \subseteq V(\H)$ such that $1 \leq |C|,|D| \leq k-1$ and $|C \cap D| \leq k-2$, and let $\ora{Q}$ be a $(C,D)$-extendable path. We assume that $A \cup V(\ora{P}) \cup B$ and $C \cup V(\ora{Q}) \cup D$ are disjoint. If $e \in E(\H)$ satisfies $e \cap (A \cup V(\ora{P}) \cup B) = B$ and $e \cap (C \cup V(\ora{Q}) \cup D) = C$, then $\ora{P} \oplus (e) \oplus \ora{Q}$ is an $(A,D)$-extendable path.
\end{mylemma}

\begin{proof}
	Write $\ora{P}=(e_1,\ldots,e_L)$ and $\ora{Q}=(e'_1,\ldots,e'_M)$, and define $\ora{R} \defeq \ora{P} \oplus (e) \oplus \ora{Q}$. Let us first check that $\ora{R}$ is a $(k-2)$-linear path. Any intersection between two edges of $\ora{R}$ is of one of four forms:
	\begin{enumerate}[noitemsep,nolistsep,label=(\arabic*)]
		\item $e_i \cap e_j$ or $e'_i \cap e'_j$.
			\\ Those are covered by the $(k-2)$-linearity of $\ora{P}$ and $\ora{Q}$ respectively.
		\item $e_i \cap e'_j$.
			\\ Those are empty because $V(\ora{P})$ and $V(\ora{Q})$ are disjoint by assumption.
		\item $e_i \cap e$ where $1 \leq i \leq L-1$ or $e \cap e'_i$ where $2 \leq i \leq M$.
			\\ By symmetry, we only address $e_i \cap e$. Since $\ora{P}$ is from $A$ to $B$, we know $e_i \cap B = \varnothing$. Moreover $e \cap V(\ora{P}) \subseteq B$ by assumption, so $e_i \cap e = \varnothing$.
		\item $e_L \cap e$ or $e \cap e'_1$.
			\\ By symmetry, we only address $e_L \cap e$. Since $\ora{P}$ is $(A,B)$-extendable, we know $|e_L \cap B| \leq k-2$, moreover the assumption on $e$ implies $e_L \cap e = e_L \cap B$ hence $|e_L \cap e| \leq k-2$.
	\end{enumerate}
	We now verify that $\ora{R}$ is from $A$ to $D$ and is $(A,D)$-extendable. By symmetry, we only show the conditions on $A$, for which we distinguish two cases:
	\begin{itemize}[noitemsep,nolistsep]
		\item If $L=0$, then the first edge of $\ora{R}$ is $e$. We have $A \cap e = A \cap B$ by the assumption on $e$, where $A \cap B \neq \varnothing$ (because $L=0$) and $|A \cap B| \leq k-2$ (by assumption), therefore $1 \leq |A \cap e| \leq k-2$. It remains to show that $A \cap e'_i = \varnothing$ for all $1 \leq i \leq M$, which is obvious since $A$ is disjoint from $V(\ora{Q})$.
		\item If $L \geq 1$, then the first edge of $\ora{R}$ is $e_1$. Since $\ora{P}$ is from $A$ to $B$, we have $A \cap e_1 \neq \varnothing$ and $A \cap e_i = \varnothing$ for all $2 \leq i \leq L$. Moreover $|A \cap e_1| \leq k-2$ because $\ora{P}$ is $(A,B)$-extendable. It remains to show that $A \cap e = \varnothing$, which is clear since $A \cap e \subseteq B$ and $A \cap B = \varnothing$ ($L \geq 1$), and that $A \cap e'_i = \varnothing$ for all $1 \leq i \leq M$, which is obvious since $A$ is disjoint from $V(\ora{Q})$. \qedhere
	\end{itemize}
\end{proof}

\begin{figure}[htbp]
	\centering
	\includegraphics[scale=.49]{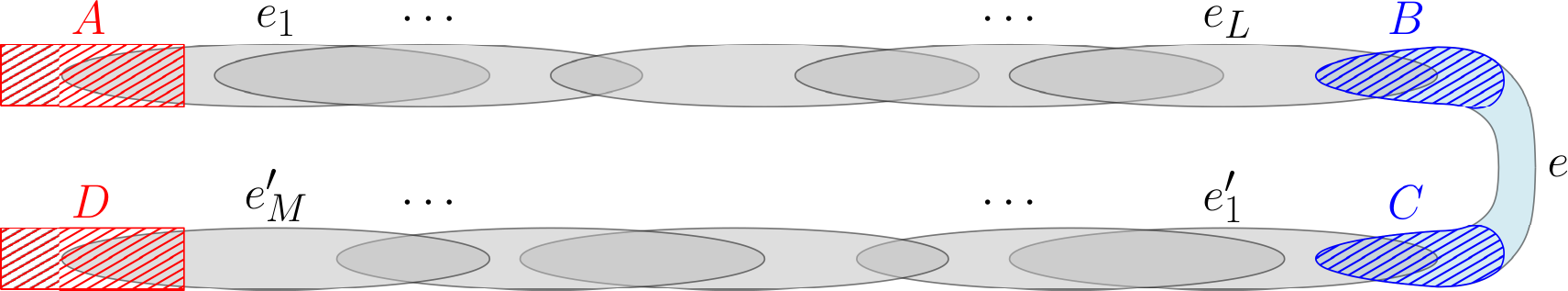}
	\caption{Illustration of \Cref{Lemma3}.}\label{ExtensionLemma3}
\end{figure}

\subsection{Archipelagos}

\hphantom{\indent}In this subsection, we fix some $x^* \in V(\H)$.

\subsubsection{Definition}

\begin{mydefinition}
	Let $\I$ and $\I'$ be disjoint islands in $\H$, where $\I'$ has an entry $\epsilon$ of size $k-1$. An edge $e \in E(\H)$ of the form $e=\{x\} \cup \epsilon$ for some $x \in V(\I)$ is called a \textit{crossing edge} from $\I$ to $\I'$. We denote by $C(\I,\I') \subseteq E(\H)$ the set of all crossing edges from $\I$ to $\I'$ in $\H$. If $\A$ is a subhypergraph of $\H$ containing $\I$ and $\I'$, we use the notation $C_{\A}(\I,\I') \defeq C(\I,\I') \cap E(\A)$.
\end{mydefinition}

\begin{myremark}
	The above definition depends on the choice of $\epsilon$ (an island might have several possible entries suiting the definition). However, we will always specify the entries when defining islands and therefore consider crossing edges for those specific entries.
\end{myremark}

\begin{mydefinition}\label[mydefinition]{def-archipelago}
	An \textit{$x^*$-archipelago} is a subhypergraph $\A$ of $\H$ such that there exist subhypergraphs $\I_1,\ldots,\I_N$ of $\A$ that are pairwise-disjoint islands with respective entries $\epsilon_1,\ldots,\epsilon_N$ satisfying the following properties:
	\begin{itemize}[noitemsep,nolistsep]
		\item $\epsilon_1=\{x^*\}$.
		\item $|\epsilon_i|=k-1$ for all $2 \leq i \leq N$.
		\item $V(\A)=V(\I_1)\cup\ldots\cup V(\I_N)$.
		\item All edges in $E(\A)\setminus(E(\I_1)\cup\ldots\cup E(\I_N))$ are crossing edges between some of the $\I_i$, such that the digraph $G$ defined by $V(G)=\{\I_1,\ldots,\I_N\}$ and $E(G)=\{(\I_i,\I_j),C_{\A}(\I_i,\I_j) \neq \varnothing\}$ contains a spanning arborescence rooted at $\I_1$. If $G$ is exactly a spanning arborescence rooted at $\I_1$, we say $\A$ is an \textit{arborescent $x^*$-archipelago}.
	\end{itemize}
	Since $x^*$ is fixed, we usually call $\A$ an \textit{archipelago} for short.
\end{mydefinition}

\begin{myremark}
	By definition of a crossing edge, there cannot exist a crossing edge from some $\I_i$ to $\I_1$ in an archipelago since $|\epsilon_1|=1 \neq k-1$. In other words, $\I_1$ has in-degree zero in $G$.
\end{myremark}

\indent Therefore, an archipelago is a union of pairwise-disjoint islands and crossing edges between some of them, satisfying specific properties. See \Cref{Archipelago} for an example (for clarity, we will use $k=3$ for all figures from now on). We will later see that an archipelago has a unique decomposition in islands, but for now we have to give ourselves islands and entries suiting the definition whenever we consider an archipelago.

\begin{figure}[htbp]
	\centering
	\includegraphics[width=1\textwidth]{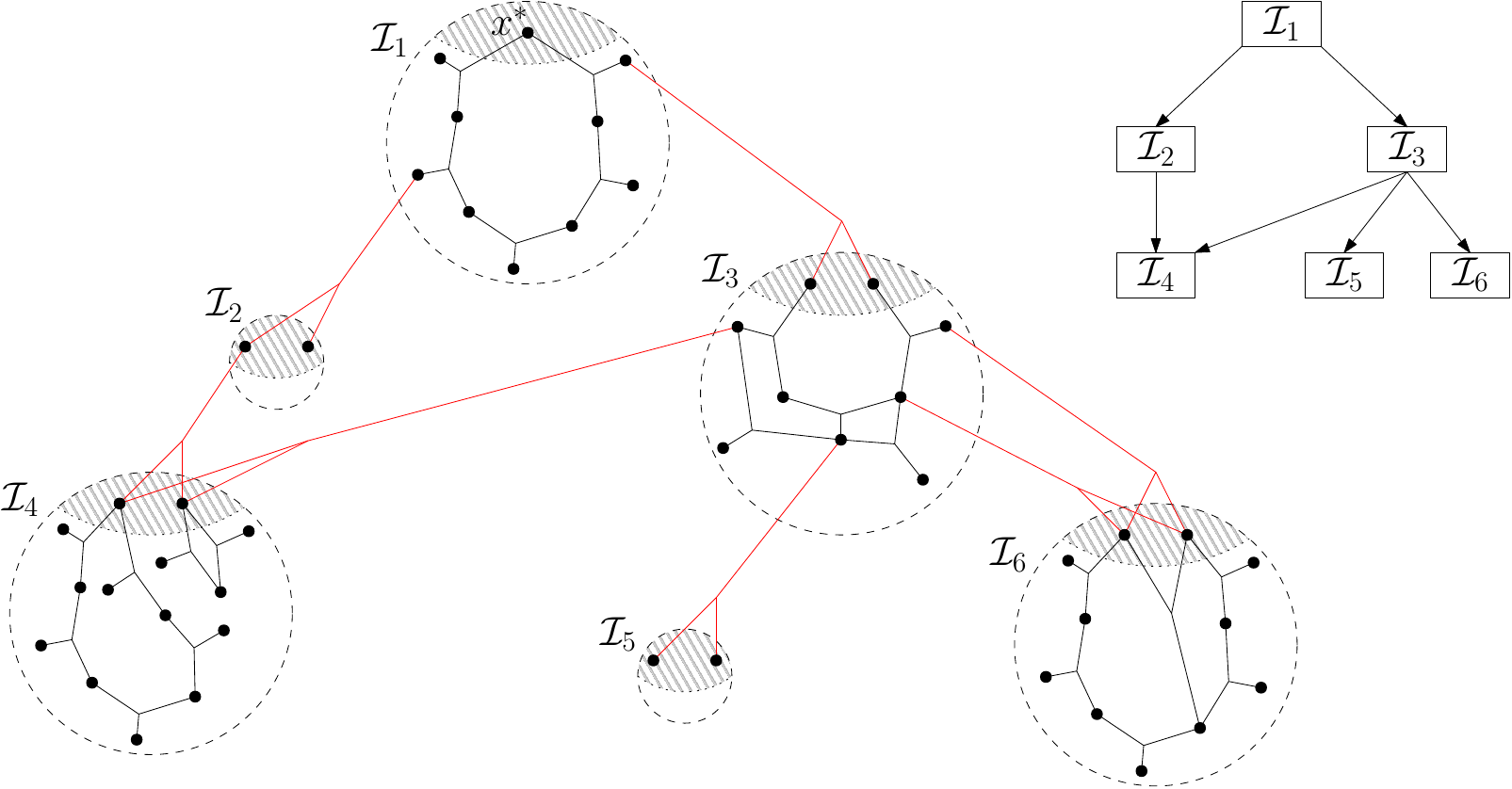}
	\caption{An archipelago which is not arborescent (with the digraph $G$ on the right). Crossing edges will always be represented in red.}\label{Archipelago}
\end{figure}

\subsubsection{Properties}

\hphantom{\indent}The next two results show how $(k-2)$-linear paths in $\A$ are related to paths in the digraph $G$. Obviously, by definition of an archipelago, a $(k-2)$-linear path in $\A$ starting from $x^*$ necessarily visits successive islands, using crossing edges to jump from one island to another. The following proposition states that, additionally, a crossing edge can only be used in one direction which is given by the digraph $G$, therefore each island is entered through its entry (hence the terminology) and it is impossible to reenter an island after leaving it.

\begin{mydefinition}
	Let $G$ be a digraph and let $v,v' \in V(G)$. A \textit{path} from $v$ to $v'$ in $G$ is a sequence denoted by $v=v_1 \to v_2 \to \ldots \to v_l=v'$ ($l \geq 1$) where $v_1,\ldots,v_l \in V(G)$ are pairwise distinct and $(v_i,v_{i+1}) \in E(G)$ for all $1 \leq i \leq l-1$.
\end{mydefinition}

\begin{myproposition}\label[myproposition]{prop-entry1}
	Let $\A$ be an archipelago, with $\I_1,\ldots,\I_N,\epsilon_1,\ldots,\epsilon_N,G$ suiting the definition. Let $\ora{P}=(e_1,\ldots,e_L)$ be a $(k-2)$-linear path from $x^*$ to some $x \in V(\I_i)$ ($1 \leq i \leq N$) in $\A$. Then the islands visited by $\ora{P}$ form a path $\I_1=\I_{i_1} \to \ldots \to \I_{i_M}=\I_i$ in $G$, and $\ora{P}$ is of the form $\ora{P}=\ora{P_1} \oplus (e_{1,2}) \oplus \ora{P_2} \oplus (e_{2,3}) \oplus \ldots \oplus \ora{P_{M-1}} \oplus (e_{M-1,M}) \oplus \ora{P_M}$ where:
	\begin{itemize}[noitemsep,nolistsep]
		\item For all $1 \leq p \leq M$: $E(\ora{P_p}) \subseteq E(\I_{i_p})$.
		\item For all $2 \leq p \leq M$: $e_{p-1,p} \in C_{\A}(\I_{i_{p-1}},\I_{i_p})$.
	\end{itemize}
	In particular, if $L \geq 1$, then for all $1 \leq p \leq M$ there is an edge of $\ora{P}$ that contains $\epsilon_{i_p}$.
\end{myproposition}

\begin{proof}
	That last assertion is clear: for $p=1$ we have $\epsilon_{i_p}=\epsilon_1=\{x^*\}\subset e_1$, and for $p \geq 2$ we have $\epsilon_{i_p} \subset e_{p-1,p}$ by definition of $C_{\A}(\I_{i_{p-1}},\I_{i_p})$. Let us now prove the main assertion.
	\\ We proceed by induction on $L$. The case $L=0$ is trivial: we have $x=x^*$ so we can set $M=1$ and $\ora{P_1}=\ora{P}=()$. Let $L \geq 1$ and assume the result to be true for all $(k-2)$-linear paths that are shorter than $\ora{P}$. The idea is to separate two simple cases: either we are currently visiting the island $\I_i$ (case $e_L \in E(\I_i)$) or we have just jumped onto $\I_i$ from another island (case $e_L \not\in E(\I_i)$).
	\\ Let $y \in e_{L-1} \cap e_L$ if $L \geq 2$, or define $y=x^*$ if $L=1$ , so that in both cases $\ora{Q}\defeq(e_1,\ldots,e_{L-1})$ is a $(k-2)$-linear path from $x^*$ to $y$ in $\A$. We have $y \in V(\I_j)$ for some $1 \leq j \leq N$. By the induction hypothesis, there exists a path $\I_1=\I_{i_1} \to \ldots \to \I_{i_M}=\I_j$ in $G$ such that we can write $\ora{Q}=\ora{Q_1} \oplus (e_{1,2}) \oplus \ora{Q_2} \oplus (e_{2,3}) \oplus \ldots \oplus \ora{Q_{M-1}} \oplus (e_{M-1,M}) \oplus \ora{Q_M}$ where $E(\ora{Q_p}) \subseteq E(\I_{i_p})$ for all $1 \leq p \leq M$ and $e_{p-1,p} \in C_{\A}(\I_{i_{p-1}},\I_{i_p})$ for all $2 \leq p \leq M$.
	\begin{figure}[htbp]
		\centering
		\includegraphics[width=1\textwidth]{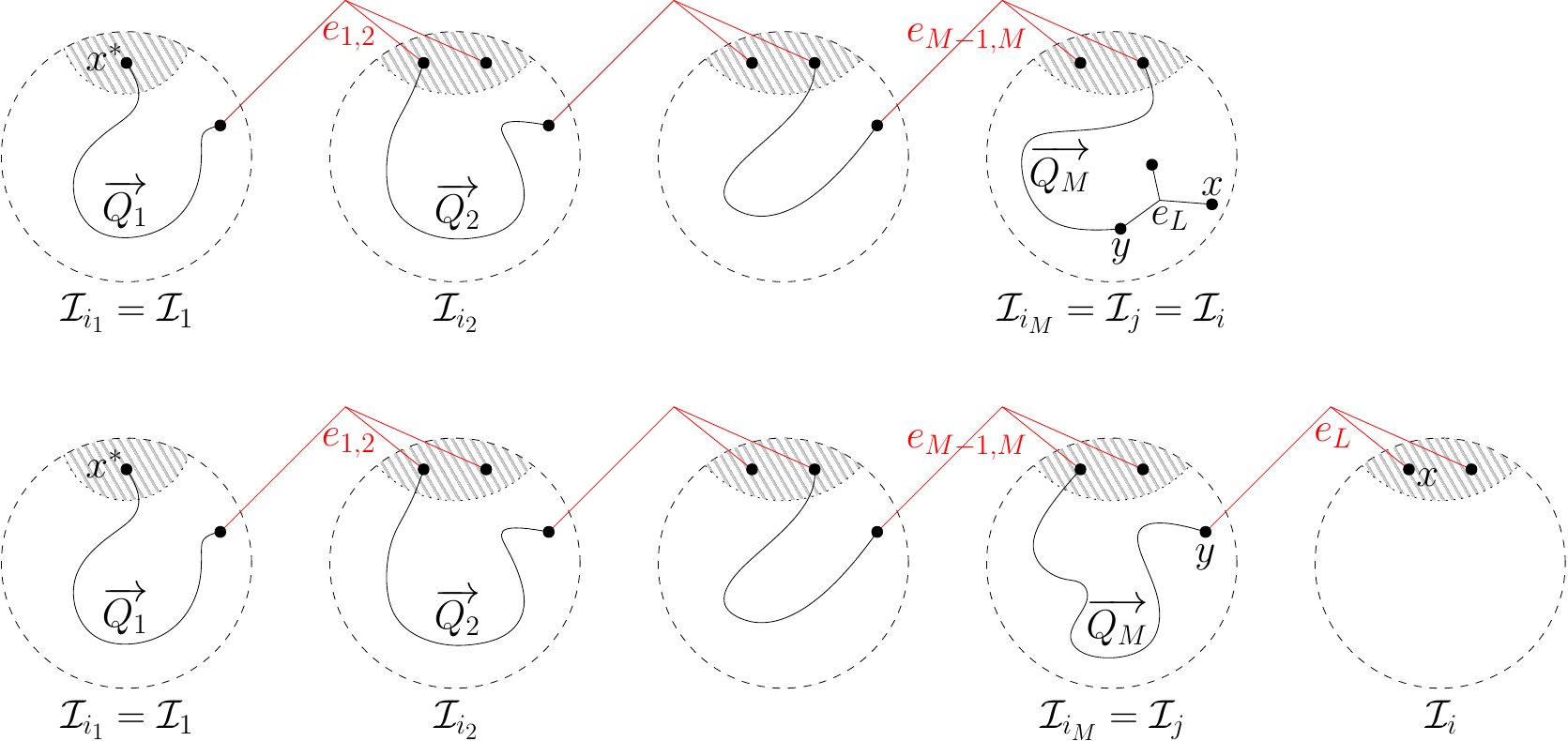}
		\caption{Top: $e_L \in E(\I_i)$. Bottom: $e_L \not\in E(\I_i)$.}\label{Path}
    \end{figure}

	\begin{itemize}
		\item First suppose that $e_L \in E(\I_i)$ (see \Cref{Path}, top). Since $y \in e_L$, this implies $i=j$, so $\ora{P_M} \defeq \ora{Q_M} \oplus (e_L)$ satisfies $E(\ora{P_M}) \subseteq E(\I_i)$. Therefore, the following writing of $\ora{P}$ completes the proof: $\ora{P}=\ora{Q} \oplus (e_L) = \ora{Q_1} \oplus (e_{1,2}) \oplus \ora{Q_2} \oplus (e_{2,3}) \oplus \ldots \oplus \ora{Q_{M-1}} \oplus (e_{M-1,M}) \oplus \ora{P_M}$.
		\item Now suppose $e_L \not\in E(\I_i)$ (see \Cref{Path}, bottom), then by definition of an archipelago we have either $e_L \in C_{\A}(\I_i,\I_j)$ or $e_L \in C_{\A}(\I_j,\I_i)$.  
			\\ Suppose for a contradiction that $e_L \in C_{\A}(\I_i,\I_j)$ i.e. $e_L=\{x\} \cup \epsilon_j$: in particular $j \neq 1$ (and $|\epsilon_j|=k-1$), so the fact that $\epsilon_j \subset e_{M-1,M}$ contradicts the $(k-2)$-linearity of $\ora{P}$ since $\epsilon_j \subset e_L$.
			\\ Therefore $e_L \in C_{\A}(\I_j,\I_i)$. In particular $i \neq 1$ (and $|\epsilon_i|=k-1$), so it is impossible that $\I_i$ has been visited before: if we had $i \in \{i_1,\ldots,i_M\}$ then some edge of $\ora{Q}$ would contain $\epsilon_i$ which would contradict the $(k-2)$-linearity of $\ora{P}$ once again. Setting $i_{M+1} \defeq i$, this ensures that the islands visited by $\ora{P}$ form a path $\I_1=\I_{i_1} \to \ldots \to \I_{i_M}=\I_j \to \I_{i_{M+1}}= \I_i$ in $G$, and we can write $\ora{P}=\ora{Q}\oplus(e_{M,M+1})\oplus\ora{P_{M+1}}$ where $e_{M,M+1} \defeq e_L \in C_{\A}(\I_{i_M},\I_{i_{M+1}})$ and $\ora{P_{M+1}} \defeq ()$, which concludes. \qedhere
	\end{itemize}
\end{proof}	

\indent Conversely, paths in $G$ yield $(k-2)$-linear paths in $\A$. The following proposition is a generalization to archipelagos of the property that defines an island.

\begin{figure}[htbp]
	\centering
	\includegraphics[scale=.58]{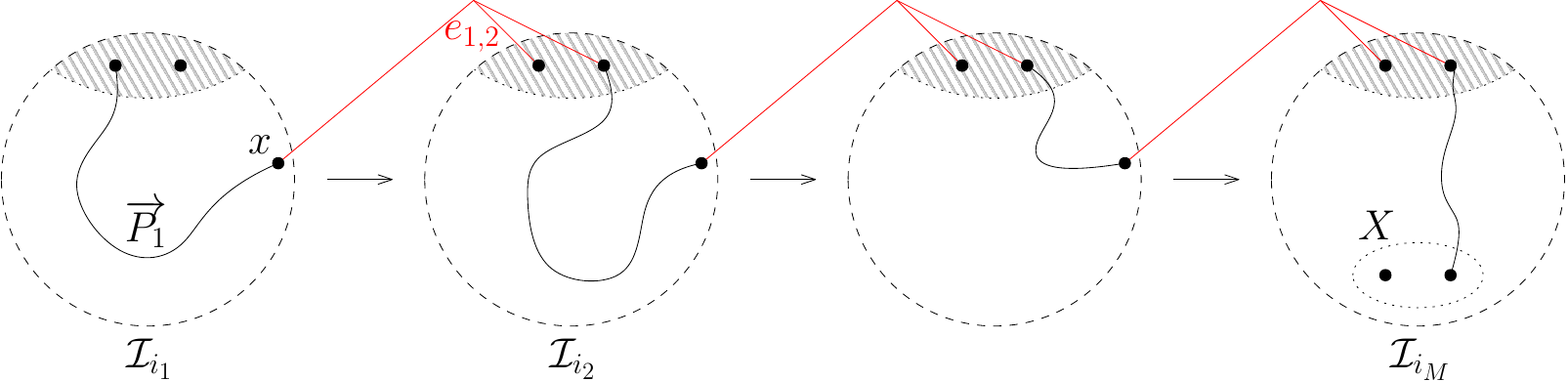}
	\caption{An $(\epsilon_{i_1},X)$-extendable path in an archipelago.}\label{RightWay}
\end{figure}

\begin{myproposition}\label[myproposition]{prop-path1}
	Let $\A$ be an archipelago, with $\I_1,\ldots,\I_N,\epsilon_1,\ldots,\epsilon_N,G$ suiting the definition. Let $X \subset V(\A)$ such that $1 \leq |X| \leq k-1$ and $X \not\in \{\epsilon_2,\ldots,\epsilon_N\}$. For all $1 \leq j \leq N$ and for every path $\I_j=\I_{i_1} \to \ldots \to \I_{i_M}$ in $G$ satisfying $X \cap V(\I_{i_M}) \neq \varnothing$ and $X \cap V(\I_{i_p}) = \varnothing$ for all $1 \leq p \leq M-1$, there exists an $(\epsilon_j,X)$-extendable path $\ora{P}$ in $\A$ of the form $\ora{P}=\ora{P_1} \oplus (e_{1,2}) \oplus \ora{P_2} \oplus (e_{2,3}) \oplus \ldots \oplus \ora{P_{M-1}} \oplus (e_{M-1,M}) \oplus \ora{P_M}$ where:
	\begin{itemize}[noitemsep,nolistsep]
		\item For all $1 \leq p \leq M$: $E(\ora{P_p}) \subseteq E(\I_{i_p})$.
		\item For all $2 \leq p \leq M$: $e_{p-1,p} \in C_{\A}(\I_{i_{p-1}},\I_{i_p})$.
	\end{itemize}
\end{myproposition}

\begin{proof}
	We proceed by induction on $M$.
	\begin{itemize}
		\item First suppose $M=1$: we need to show that if $X \cap V(\I_j) \neq \varnothing$ then there exists an $(\epsilon_j,X)$-extendable path in $\I_j$. This is basically the definition of an island, except that $X$ is not necessarily entirely included in $V(\I_j)$. This is not a problem: since $X \not\in \{\epsilon_2,\ldots,\epsilon_N\}$ by assumption, there exists an $(\epsilon_j,X \cap V(\I_j))$-extendable path $\ora{P}$ in $\I_j$ by definition of an island, and $\ora{P}$ is also $(\epsilon_j,X)$-extendable by \Cref{Lemma1}.
		\item Now suppose $M \geq 2$ and assume the result to be true for all shorter paths in $G$. We build the desired $(\epsilon_{i_1},X)$-extendable path by assembling three parts:
			\begin{enumerate}[label=(\arabic*)]
				\item By the induction hypothesis, there exists an $(\epsilon_{i_2},X)$-extendable path $\ora{P'}$ in $\A$ of the form $\ora{P'}=\ora{P_2} \oplus (e_{2,3}) \oplus \ora{P_3} \oplus (e_{3,4}) \oplus \ldots \oplus \ora{P_{M-1}} \oplus (e_{M-1,M}) \oplus \ora{P_M}$ where $E(\ora{P_p}) \subseteq E(\I_{i_p})$ for all $2 \leq p \leq M$ and $e_{p-1,p} \in C_{\A}(\I_{i_{p-1}},\I_{i_p})$ for all $3 \leq p \leq M$.
				\item Let $e_{1,2} \in C_{\A}(\I_{i_1},\I_{i_2})$, which exists since $(\I_{i_1},\I_{i_2}) \in E(G)$: we have $e_{1,2}=\{x\} \cup \epsilon_{i_2}$ for some $x \in V(\I_{i_1})$.
				\item Finally, by definition of an island, there exists an $(\epsilon_{i_1},x)$-extendable path $\ora{P_1}$ in $\I_{i_1}$.
			\end{enumerate}
			The path $\ora{P} \defeq \ora{P_1} \oplus (e_{1,2}) \oplus \ora{P'}$ is represented in \Cref{RightWay}. \Cref{Lemma3} applied to $A=\epsilon_{i_1}$, $B=\{x\}$, $C=\epsilon_{i_2}$ and $D=X$ ensures that $\ora{P}$ is an $(\epsilon_{i_1},X)$-extendable path. \qedhere
	\end{itemize}
\end{proof}

\indent We get the following characterization for the entries of an archipelago:

\begin{myproposition}\label[myproposition]{prop-path2}
	Let $\A$ be an archipelago, with $\I_1,\ldots,\I_N,\epsilon_1,\ldots,\epsilon_N$ suiting the definition. Let $X \subseteq V(\A)$ such that $1 \leq |X| \leq k-1$. There exists an $(x^*,X)$-extendable path in $\A$ if and only if $X \not\in \{\epsilon_2,\ldots,\epsilon_N\}$.
\end{myproposition}

\begin{proof}
	We distinguish both cases:
	\begin{itemize}[noitemsep,nolistsep]
		\item Suppose $X=\epsilon_i$ for some $2 \leq i \leq N$. Let $\ora{P}$ be a $(k-2)$-linear path from $x^*$ to $\epsilon_i$ in $\A$, then $\ora{P}$ is a $(k-2)$-linear path from $x^*$ to $x$ in $\A$ for some $x \in \epsilon_i$. By \Cref{prop-entry1}, some edge of $\ora{P}$ (necessarily the last one, since $\ora{P}$ is from $x^*$ to $\epsilon_i$) contains $\epsilon_i$, which proves that $\ora{P}$ is not $(x^*,\epsilon_i)$-extendable.
		\item Suppose $X \not\in \{\epsilon_2,\ldots,\epsilon_N\}$. Out of all the paths in $G$ from $\I_1$ to one of the islands intersecting $X$ (recall that $G$ contains a spanning arborescence rooted at $\I_1$, so there exists at least one), consider a shortest one, so that $X$ only intersects the last island of that path. We can now apply \Cref{prop-path1}: there exists an $(\epsilon_1,X)$-extendable path in $\A$, which concludes since $\epsilon_1=\{x^*\}$. \qedhere
	\end{itemize}
\end{proof}

\begin{mycorollary}\label[mycorollary]{coro-inclusion}
	Let $\A$ be an archipelago in $\H$. For all $x \in V(\A)$, there exists a $(k-2)$-linear path from $x^*$ to $x$ in $\A$. In particular, $V(\A) \subseteq LCC^{\,k-2}_{\H}(x^*)$.
\end{mycorollary}

\begin{proof}
	Let $x \in V(\A)$: applying \Cref{prop-path2} to $X=\{x\}$ shows that there exists a $(k-2)$-linear path from $x^*$ to $x$ in $\A$.
\end{proof}

\indent Finally, we show that an archipelago has a unique decomposition.

\begin{myproposition}\label[myproposition]{prop-unicity}
	Any archipelago $\A$ has unique islands and entries suiting the definition.
\end{myproposition}

\begin{proof}
	Let $\epsilon_1,\ldots,\epsilon_N$ be entries suiting the definition: we have $\epsilon_1=\{x^*\}$, moreover $\{\epsilon_2,\ldots,\epsilon_N\}$ is exactly the set of all subsets $X \subset V(\A)$ such that $1 \leq |X| \leq k-1$ and there exists no $(x^*,X)$-extendable path in $\A$ by \Cref{prop-path2}, so these entries are unique. Suppose for a contradiction that $\{\I_1,\ldots,\I_N\}$ and $\{\I'_1,\ldots,\I'_N\}$ are two distinct sets of islands suiting the definition, where $\I_i$ and $\I'_i$ have the same entry $\epsilon_i$ for all $1 \leq i \leq N$. Since islands are induced subhypergraphs of $\A$, $\{V(\I_1),\ldots,V(\I_N)\}$ and $\{V(\I'_1),\ldots,V(\I'_N)\}$ are two distinct partitions of $V(\A)$, so there exists $1 \leq i \neq j \leq N$ such that $V(\I_i) \cap V(\I'_j) \neq \varnothing$. Let $x \in V(\I_i) \cap V(\I'_j)$.
	\begin{itemize}[noitemsep,nolistsep]
		\item Using the first decomposition, there exists an $(\epsilon_i,x)$-extendable path $\ora{P}=(e_1,\ldots,e_L)$ in $\I_i$ by definition of an island. For all $2 \leq l \leq N$, no edge of $\ora{P}$ contains $\epsilon_l$: if $l=i$ then this is the definition of an $(\epsilon_i,x)$-extendable path, and if $l \neq i$ then this is obvious since $V(\I_i)$ is disjoint from $\epsilon_l$.
		\item Using the second decomposition, since $x \in V(\I'_j)$ and $\epsilon_i$ is disjoint from $V(\I'_j)$, we can define $r \defeq \inf\{1 \leq p \leq L \,\,\text{such that $e_p \not\subset V(\I'_j)$} \}$. We have $e_r \not\subset V(\I'_j)$, however $e_r$ intersects $V(\I'_j)$ by minimality of $r$, therefore $e_r$ is necessarily a crossing edge for the second decomposition. This means that $\epsilon_l \subset e_r$ for some $2 \leq l \leq N$, which contradicts what we have just established. \qedhere
	\end{itemize}
\end{proof}

\begin{mynotation}
	Let $\A$ be an archipelago. \Cref{prop-unicity} allows us to define without ambiguity:
	\begin{itemize}[noitemsep,nolistsep]
		\item $\I(\A)$: the set of islands of $\A$.
		\item $\epsilon(\A)$: the set of entries of the islands of $\A$.
		\item $G(\A)$: the digraph from the definition of an archipelago.
	\end{itemize}
\end{mynotation}

\section{$(k-2)$-linear connected components: structure and computation}\label{Section3}

\hphantom{\indent}In this section, we suppose again that $\H$ is $k$-uniform and we fix some $x^* \in V(\H)$.

\subsection{Main results}

\hphantom{\indent}Our two main results about $(k-2)$-linear connected components, one structural and the other algorithmic, can be assembled into the following main theorem which will be proven in this section.

\begin{mydefinition}
	An $x^*$-archipelago $\A$ in $\H$ is said to be maximal if there is no $x^*$-archipelago in $\H$ that has $\A$ as a strict subhypergraph.
\end{mydefinition}

\begin{mytheorem}\label[mytheorem]{theo-main}
	$\H[LCC^{\,k-2}_{\H}(x^*)]$ is the unique maximal $x^*$-archipelago in $\H$, and it can be computed in $O(m^2 k)$ time where $m=|E(\H)|$.
\end{mytheorem}

\begin{mycorollary}\label[mycorollary]{coro-main}
	For all $k \geq 3$, $\textsc{HypConnectivity}_{k,k-2}$ is solvable in polynomial time.
\end{mycorollary}

\subsection{The key intermediate result}

\hphantom{\indent}Theorem \ref{theo-main} will come as a straightforward consequence of the following theorem, which is illustrated in \Cref{Partition}:

\begin{mytheorem}\label[mytheorem]{theo-algo}
	 There exists an $x^*$-archipelago $\A$ in $\H$ and a partition $E(\H)=E(\A) \cup E_{cut} \cup E_{ext}$ (where $E_{cut}$ and/or $E_{ext}$ may be empty) such that:
	\begin{enumerate}[noitemsep,nolistsep,label=(\arabic*)]
		\item Every $e \in E_{cut}$ is of the form $e=\epsilon \cup \{x\}$ for some entry $\epsilon$ of $\A$ of size $k-1$ and some $x \not\in V(\A)$;
		\item Every $e \in E_{ext}$ is disjoint from $V(\A)$.
	\end{enumerate}
	Moreover, this partition can be computed in $O(m^2 k)$ time where $m=|E(\H)|$.
\end{mytheorem}

\begin{figure}[htbp]
	\centering
	\includegraphics[width=1\textwidth]{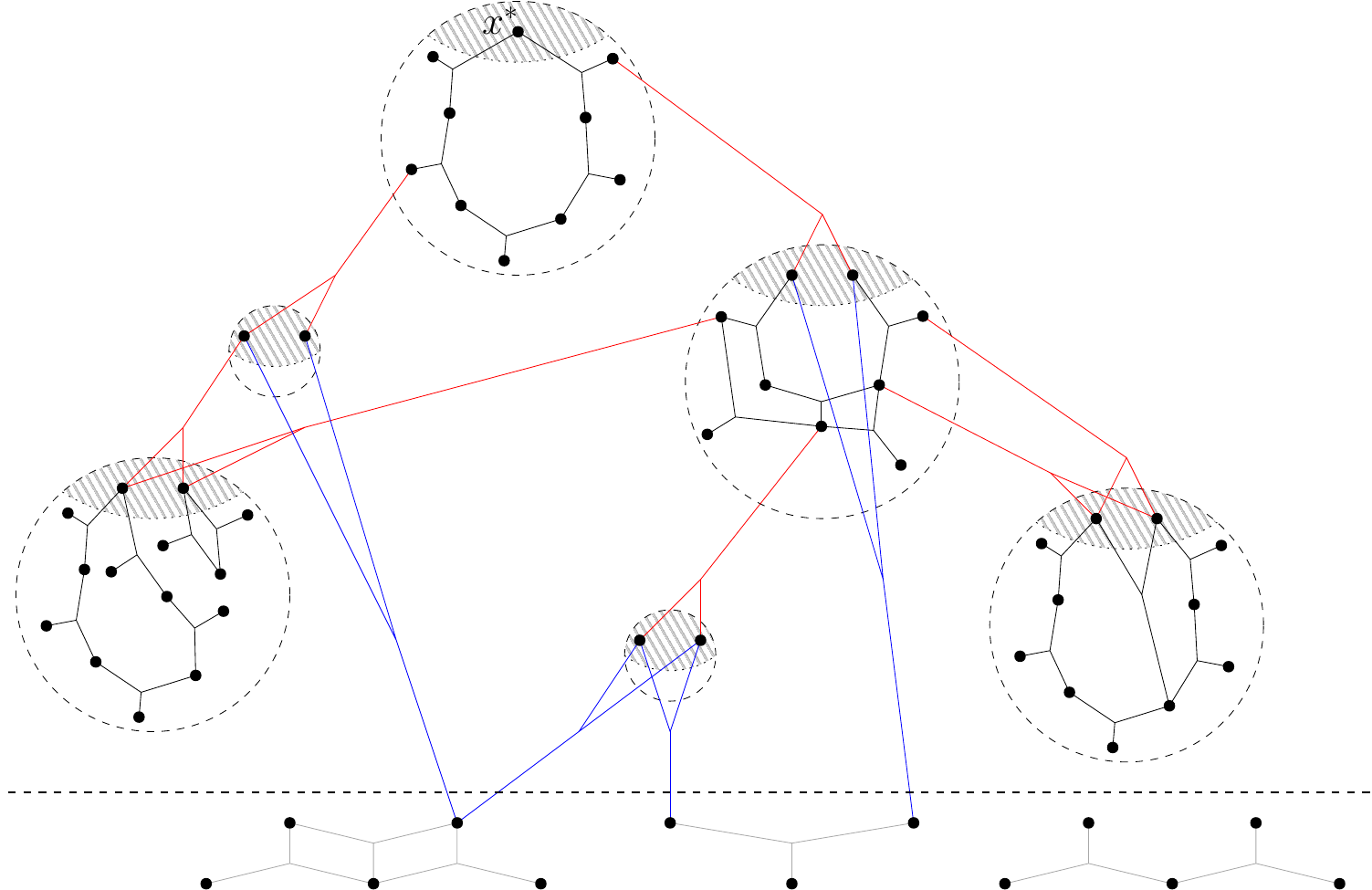}
	\caption{The hypergraph $\H$ is represented in full. In black and red: $E(\A)$ (archipelago). In blue: $E_{cut}$. In grey, below the dashed line: $E_{ext}$.}\label{Partition}
\end{figure}

\subsubsection{Augmenting archipelagos}

\hphantom{\indent}Our algorithm proving \Cref{theo-algo} will build the archipelago $\A=\H[LCC^{\,k-2}_{\H}(x^*)]$ edge by edge until reaching maximality, and then throw the remaining edges into $E_{cut}$ and $E_{ext}$. Therefore, we need to address the following question: given an archipelago $\A$ and an edge $e \in E(\H) \setminus E(\A)$, is $\A \cup e$ an archipelago (and if so, for what decomposition)? Here $\A \cup e$ denotes the subhypergraph of $\H$ defined by $V(\A \cup e)=V(A) \cup e$ and $E(\A \cup e)=E(A) \cup \{e\}$. The answer will depend on the way $e$ intersects $\A$:

\begin{mydefinition}\label{def-types}
Let $\A$ be an archipelago. An edge $e \in E(\H) \setminus E(\A)$ is of one of five \textit{$\A$-types}:
	\begin{enumerate}[noitemsep,nolistsep,label=\arabic*.]
		\item "exterior": $|e \cap V(\A)|=0$.
		\item "new crossing": $|e \cap V(\A)|=1$.
		\item "crossing": $e$ is a crossing edge between two islands of $\A$.
		\item "cut": $e$ is of the form $e=\epsilon \cup \{x\}$, where $\epsilon$ is an entry of $\A$ of size $k-1$ and $x \in V(\H)\setminus V(\A)$.
		\item "other": $e$ is none of the above.
	\end{enumerate}
	Those are well defined because the islands and entries of an archipelago are unique by \Cref{prop-unicity}. The five $\A$-types are illustrated in \Cref{Types}.
\end{mydefinition}

\begin{figure}[htbp]
	\centering
	\includegraphics[scale=.6]{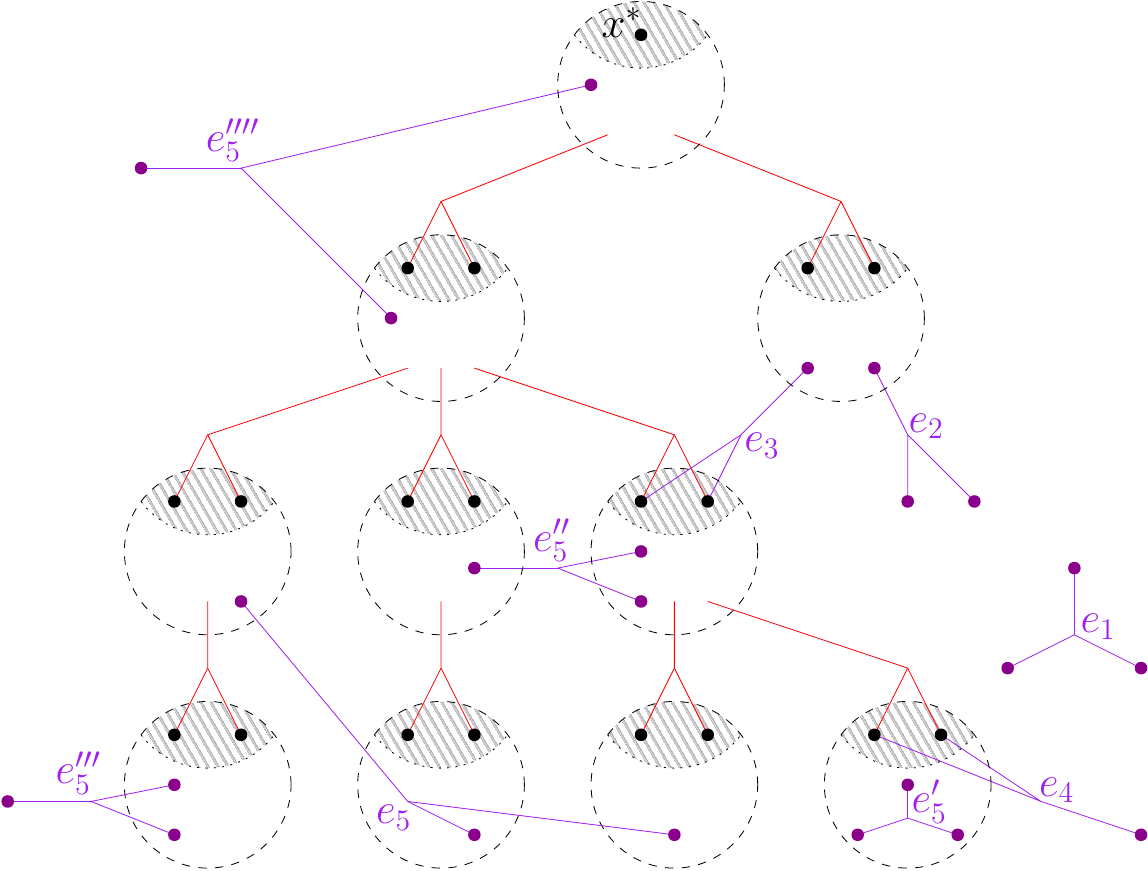}
	\caption{An arborescent archipelago $\A$ (the inside of the islands is not detailed), and some edges in $E(\H) \setminus E(\A)$ (in purple). The names of the edges follow the numbering from \Cref{def-types}: $e_1$ is of $\A$-type "exterior", $e_2$ is of $\A$-type "new crossing", etc.}\label{Types}
\end{figure}

\indent Fundamentally:
\begin{itemize}[noitemsep,nolistsep]
	\item The $\A$-types "crossing", "new crossing" and "other" correspond to edges that get added to the archipelago.
	\item The $\A$-type "cut" corresponds to $E_{cut}$.
	\item The $\A$-type "exterior" corresponds to $E_{ext}$.
\end{itemize}
\indent Let $\A$ be an archipelago, with islands $\I_1,\ldots,\I_N$ and entries $\epsilon_1,\ldots,\epsilon_N$, and let $e \in E(\H) \setminus E(\A)$. We now explain why $\A \cup e$ is an archipelago if $e$ is of $\A$-type "crossing", "new crossing" or "other". In the case of the $\A$-types "new crossing" and "other", the arborescent nature of the archipelago will be preserved, so those edges will be added first in our algorithm so that the archipelago remains arborescent for as long as possible. Even though the decomposition of $\A \cup e$ is given by $\I(A \cup e)$ and $\epsilon(A \cup e)$ alone, we also describe $G(\A \cup e)$ in the arborescent case.

\paragraph{$\quad$ I) $e$ is of $\A$-type "new crossing"}\mbox{}\\

\indent This case is easy: a new island is created, with $e$ being the crossing edge that connects it to the rest (see \Cref{Type1a}).

\begin{figure}[htbp]
	\centering
	\includegraphics[width=1\textwidth]{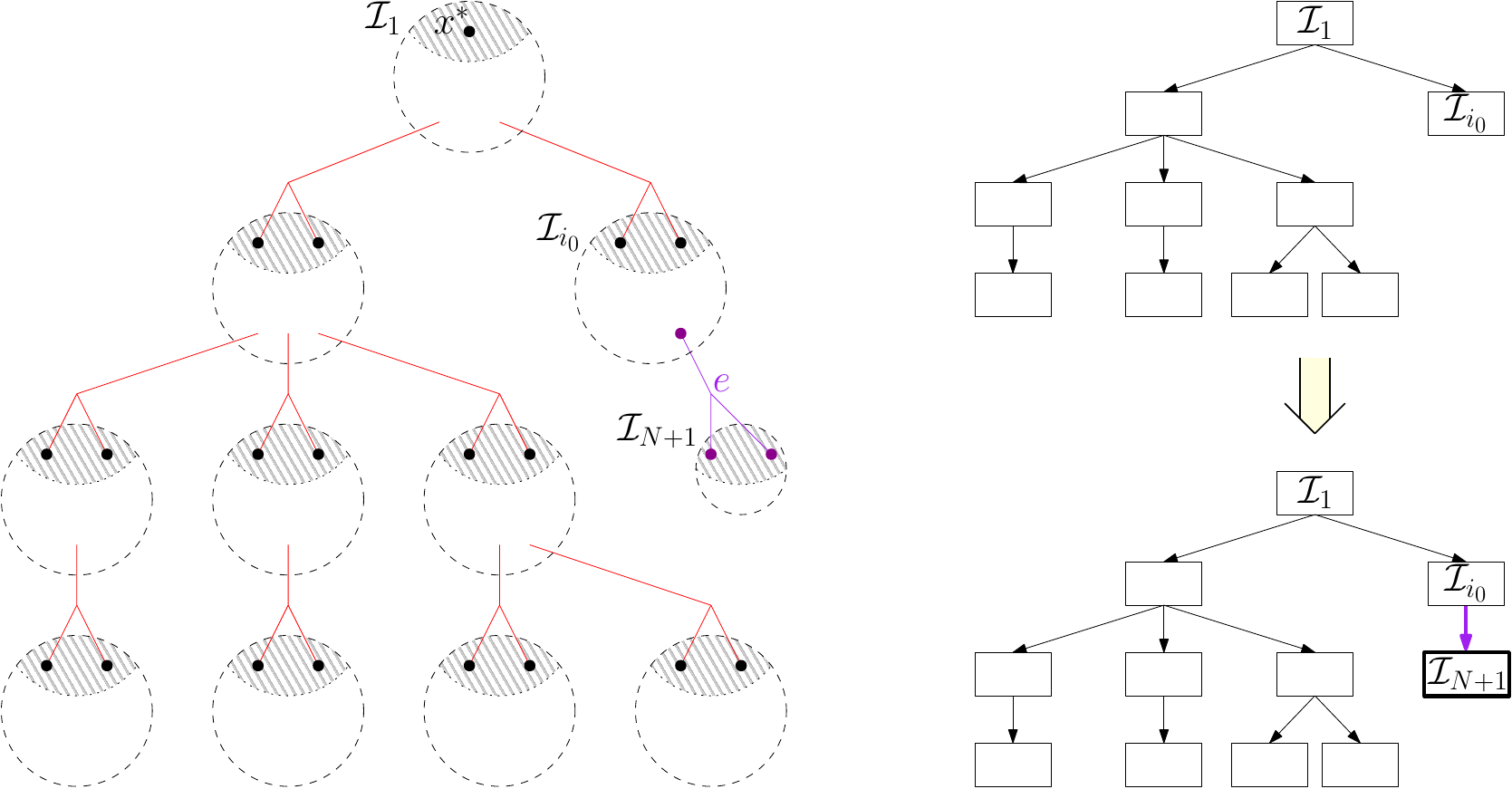}
	\caption{The archipelago $\A \cup e$ where $\A$ is as in \Cref{Types} and $e=e_2$. On the right: the arborescences $G(\A)$ (top) and $G(\A \cup e)$ (bottom).}\label{Type1a}
\end{figure}

\begin{myproposition}\label[myproposition]{prop-type1a}
	Suppose $\A$ is arborescent and $e$ is of $\A$-type "new crossing". Let $1 \leq i_0 \leq N$ be the index of the only island that intersects $e$, and let $\I_{N+1}$ be the empty island with entry $\epsilon_{N+1} \defeq e \setminus V(\I_{i_0})$. Then $\A \cup e$ is an arborescent archipelago with:
	\begin{itemize}[noitemsep,nolistsep]
		\item $\I(A \cup e)=\I(A) \cup \{\I_{N+1}\}$.
		\item $\epsilon(\A \cup e)= \epsilon(\A) \cup \{\epsilon_{N+1}\}$.
		\item $G(\A \cup e)$ defined as the digraph obtained from $G(\A)$ by adding a new vertex $\I_{N+1}$ and an arc $(\I_{i_0},\I_{N+1})$.
	\end{itemize}
\end{myproposition}

\begin{proof}
	This is clear: $e$ is a crossing edge from $\I_{i_0}$ to $\I_{N+1}$, hence the new arc in $G(\A \cup e)$ which is obviously an arborescence since $G(\A)$ is.
\end{proof}

\paragraph{$\quad$ II) $e$ is of $\A$-type "other"}\mbox{}\\

\indent By definition, this means that: $|e \cap V(\A)| \geq 2$, $e$ is not a crossing edge, and $e$ is not of the form $\epsilon \cup \{x\}$ where $\epsilon$ is an entry of $\A$ of size $k-1$ and $x \in V(\H)\setminus V(\A)$.
\\ \indent This case is more complicated. Consider \Cref{Types}. If $e$ only intersects one island ($e=e_5'$ or $e=e_5'''$ for instance), then it should be easy to show that this island plus $e$ is still an island. If $e$ links several islands however, then the way to redefine islands is not as straightforward, since $e$ is not a crossing edge. Suppose $e=e_5$ for instance, as in \Cref{Type1b}. The fact that $e$ acts as a bridge between several islands creates new paths: for example, we have an $(x^*,\epsilon_6)$-extendable path in $\A \cup e$ (represented schematically in \Cref{Type1b}), therefore $\epsilon_6$ would not be an entry of $\A \cup e$ (recall \Cref{prop-path2}). Actually, it can be shown that the subhypergraph $\I$, formed by the union of $\I_2,\I_4,\I_5,\I_6,\I_8,\I_9$ and the crossing edges between them as well as $e$, is an island with entry $\epsilon_2$. Therefore, $\A \cup e$ is an archipelago with five islands: $\I_1,\I_3,\I_7,\I_{10},\I$. On this example, we see how adding en edge can merge islands together. We are now going to generalize this argument.

\begin{figure}[htbp]
	\centering
	\includegraphics[width=1\textwidth]{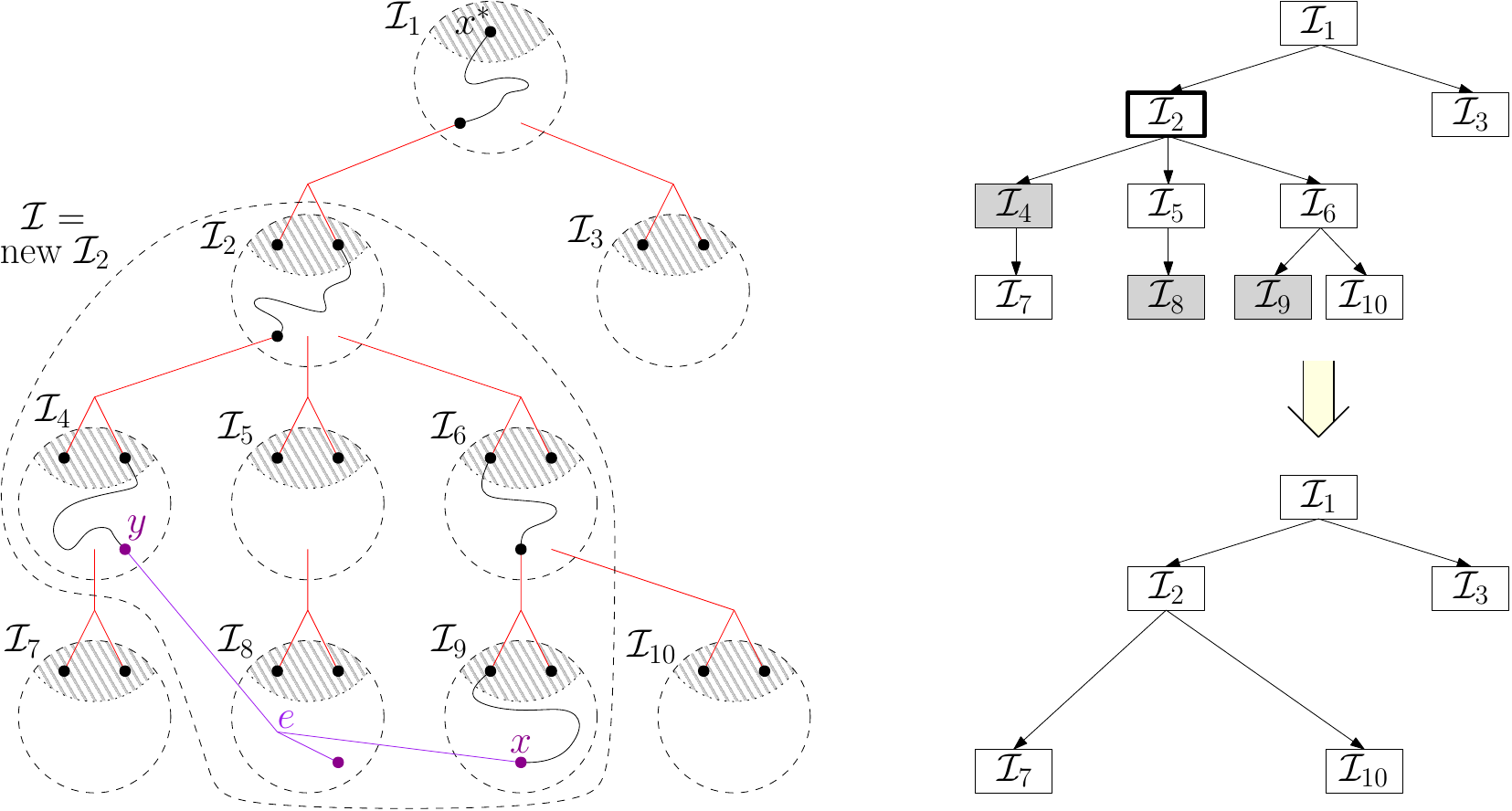}
	\caption{The archipelago $\A \cup e$ where $\A$ is as in \Cref{Types} and $e=e_5$. On the right: the arborescences $G(\A)$ (top) and $G(\A \cup e)$ (bottom).}\label{Type1b}
\end{figure}

\begin{mydefinition}
	Let $G$ be an arborescence rooted at some $v^* \in V(G)$, and let $U=\{v_1,\ldots,v_r\} \subseteq V(G)$. For all $1 \leq i \leq r$, let $v^*=v_{i,1} \to \ldots \to v_{i,l_i}=v_i$ be the unique path from $v^*$ to $v_i$ in $G$. Define $i_0 \defeq \sup\{1 \leq p \leq \min_{1 \leq i \leq r}l_i \,\,|\,\, v_{1,p}=\ldots=v_{r,p}\}$. The \textit{lowest common ancestor of $U$ in $G$} is defined as $\FCA_G(U)\defeq v_{i_0}$.
\end{mydefinition}

\begin{mydefinition}
	Let $G$ be an arborescence and let $v\in V(G)$. For all $1 \leq i \leq r$, let $v=v_{i,1} \to \ldots \to v_{i,l_i}=v_i$ be a path from $v$ to some $v_i \in V(G)$ in $G$. Let $U\defeq\bigcup_{1 \leq i \leq r}\{v_{i,1},\ldots,v_{i,l_i}\}$ be the set of all vertices on these paths. \textit{Merging $U$ into $v$} means:
	\begin{itemize}[noitemsep,nolistsep]
		\item deleting all vertices in $U \setminus \{v\}$;
		\item deleting all arcs between vertices in $U$;
		\item replacing every arc $(u,w) \in (U \setminus \{v\}) \times (V(G) \setminus U)$ by an arc $(v,w)$.
	\end{itemize}
\end{mydefinition}

\begin{myexample}
	\Cref{Type1b} features a merging process on the right. The three considered paths are: $\I_2 \leftarrow \I_4$, $\I_2 \leftarrow \I_5 \leftarrow \I_8$, $\I_2 \leftarrow \I_6 \leftarrow \I_9$. The set $U=\{\I_2,\I_4,\I_5,\I_6,\I_8,\I_9\}$ has been merged into $v=\I_2$.
\end{myexample}

\begin{myproposition}\label[myproposition]{prop-type1b}
	Suppose $\A$ is arborescent and $e$ is of $\A$-type "other". Define:
	\begin{itemize}[noitemsep,nolistsep]
		\item $J_0 \defeq \{1 \leq i \leq N \,\,\text{such that $V(\I_i) \cap e \neq \varnothing$} \}$, the set of indices of the islands that $e$ intersects. 
		\item $i_0$ the index such that $\I_{i_0}\defeq \FCA_{G(\A)}(\{\I_i,i \in J_0\})$.
		\item $J \defeq \bigcup_{i \in J_0} \{1 \leq j \leq N \,\,|\,\, \I_j \text{ is on the path from $\I_{i_0}$ to $\I_i$ in $G(\A)$} \} \supseteq J_0$.
		\item $\I \defeq \A[\bigcup_{j \in J}V(\I_j)]$, the island that will replace $\I_{i_0}$ (with the same entry $\epsilon_{i_0}$).
	\end{itemize}
	Then $\A \cup e$ is an arborescent archipelago with:
	\begin{itemize}[noitemsep,nolistsep]
		\item $\I(\A \cup e)= (\I(\A) \setminus \{\I_j, j\in J\}) \cup \{\I\}$.
		\item $\epsilon(\A \cup e)= \epsilon(\A) \setminus \{\epsilon_j, j\in J \setminus \{i_0\}\}$.
		\item $G(\A \cup e)$ defined as the digraph obtained from $G(\A)$ by merging $\{\I_j, j\in J\}$ into $\I_{i_0}$.
	\end{itemize}
\end{myproposition}

\begin{proof}
	For visual help, refer to \Cref{Type1b}: in this example we have $J_0=\{4,8,9\}$, $i_0=2$, $J=\{2,4,5,6,8,9\}$. The merging process that defines $G(\A \cup e)$ clearly preserves the fact that the digraph is an arborescence. To complete the proof, we only need to show that $\I$ is an island with entry $\epsilon_{i_0}$: let $X \subset V(\I)$ such that $1 \leq |X| \leq k-1$ (and $X \neq \epsilon_{i_0}$ if $i_0 \neq 1$), we need to find an $(\epsilon_{i_0},X)$-extendable path in $\I$. As visible in \Cref{Types}, $e$ might or might not be included in $V(\A)$, so in general we have $V(\I)=\bigcup_{j \in J}V(\I_j) \cup e$. We distinguish four possibilities:
	
	\begin{enumerate}[label=\arabic*)]
	
		\item \underline{Case 1}: $X \subset \bigcup_{j \in J}V(\I_j)$ and $X \not\in\{\epsilon_j, j\in J \setminus \{i_0\}\}$.
			\\ Of all paths in $G(\A)$ from $\I_{i_0}$ to an island intersecting $X$, let $\I_{i_0}=\I_{j_1} \to \ldots \to \I_{j_M}$ be a shortest one, so that $X \cap V(\I_{j_M})\neq \varnothing$ and $X \cap V(\I_{j_p})= \varnothing$ for all $1 \leq p \leq M-1$. Note that, by definition of $J$, we have $\{j_1,\ldots,j_M\} \subseteq J$, so the islands $\I_{j_1},\ldots,\I_{j_M}$ are all subhypergraphs of $\I$ and all crossing edges between them in $\A$ are edges of $\I$. By \Cref{prop-path1}, there exists an $(\epsilon_{i_0},X)$-extendable path $\ora{P}$ in $\A$ such that $E(\ora{P}) \subseteq \bigcup_{p=1}^M E(\I_{j_p}) \cup \bigcup_{p=2}^M C_{\A}(\I_{j_{p-1}},\I_{j_p}) \subseteq E(\I)$, which concludes.
			
		\item \underline{Case 2}: $X$ intersects both $\bigcup_{j \in J}V(\I_j)$ and $e \setminus \bigcup_{j \in J}V(\I_j)$.
			\\ Define $X' \defeq X \cap \bigcup_{j \in J}V(\I_j)$, we have $1 \leq |X'| \leq k-1$. Case 1 applied to $X'$ gives us an $(\epsilon_{i_0},X')$-extendable path $\ora{P}$ in $\I$, which is also $(\epsilon_{i_0},X)$-extendable by \Cref{Lemma1} applied to $A=\epsilon_{i_0}$, $B=X'$ and $B'=X$.
			
		\item \underline{Case 3}: $X \subset e \setminus \bigcup_{j \in J}V(\I_j)$.
			\\ Define $X' \defeq e \cap \bigcup_{j \in J}V(\I_j)$, we have $2 \leq |X'| \leq k-1$ hence $1 \leq |X| \leq k-2$: indeed $|X'| \geq 2$ by definition of the $\A$-type "other", and $|X'| \leq k-1$ because $e \setminus \bigcup_{j \in J}V(\I_j) \supseteq X \neq \varnothing$. Moreover $X' \not\in\{\epsilon_j, j\in J \setminus \{i_0\}\}$, otherwise $e$ would be of $\A$-type "cut". We can thus apply Case 1 to $X'$, which gives us an $(\epsilon_{i_0},X')$-extendable path $\ora{P}$ in $\I$. \Cref{Lemma3} applied to $A=\epsilon_{i_0}$, $B=X'$, $C=D=X$ and $\ora{Q}=()$ ensures that $\ora{P} \oplus (e)$ is an $(\epsilon_{i_0},X)$-extendable path in $\I$.
			
		\item \underline{Case 4}: $X=\epsilon_j$ for some $j \in J \setminus \{i_0\}$.
			\\ In particular $|J| \geq 2$, so $e$ intersects several islands. Note that, since $\I_{i_0}$ is a strict ancestor of $\I_j$ in $G(\A)$, we have $j \neq 1$. Remember our example from \Cref{Type1b}: we considered $X=\epsilon_6$, and the $(\epsilon_2,X)$-extendable path was obtained by going from $\epsilon_2$ to $e \cap V(\I_4)=\{y\}$, then using $e$ to jump from $\I_4$ to $\I_9$, then going from $e \cap V(\I_9)=\{x\}$ to $X$. Let us now build this path in general.
			
			\begin{figure}[htbp]
				\centering
				\includegraphics[scale=.6]{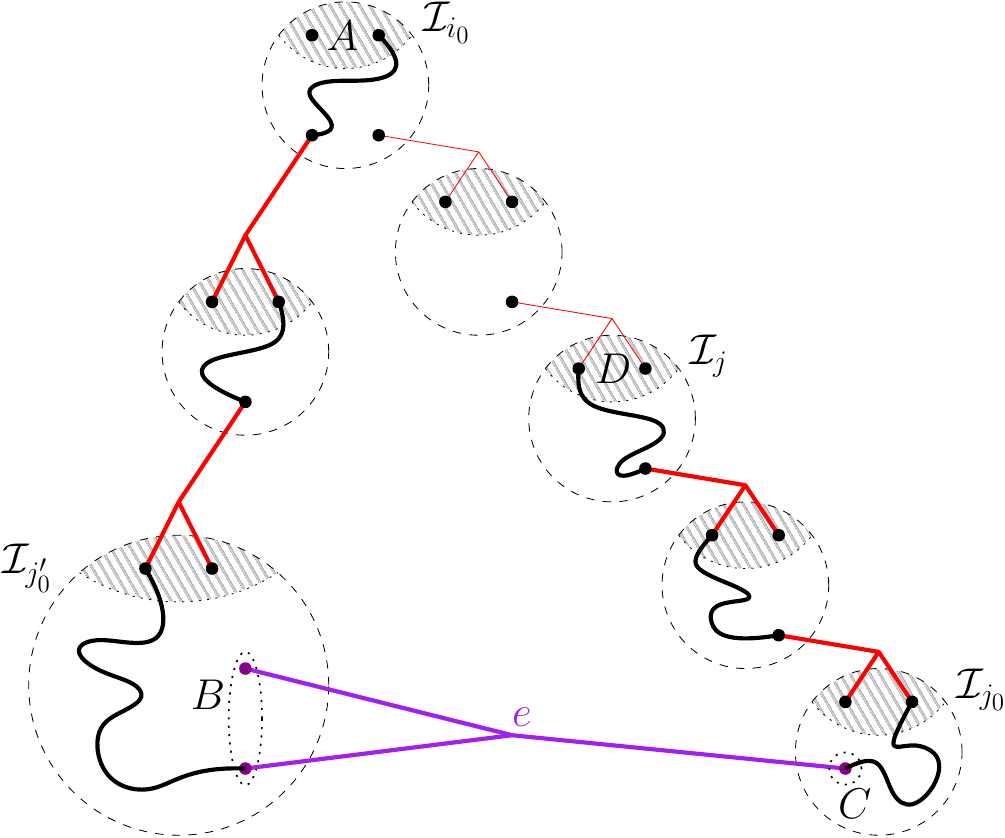}
				\caption{Illustration of Case 4 from \Cref{prop-type1b}. The bold paths (in red and black) are $\protect\ora{P}$ on the right and $\protect\ora{Q}$ on the left.}\label{Type1b_Case4}
			\end{figure}
			
			\begin{itemize}
				\item Let $j_0 \in J_0$ such that the path $\I_j=\I_{i_1} \to \ldots \to \I_{i_M}=\I_{j_0}$ in $G(\A)$ is shortest, so that $i_p \not\in J_0$ for all $1 \leq p \leq M-1$. This means $e \cap V(\I_{i_M})\neq \varnothing$ and $e \cap V(\I_{i_p})= \varnothing$ for all $1 \leq p \leq M-1$. Since $e$ intersects several islands, we know $1 \leq |e \cap V(\I_{j_0})| \leq k-1$. Moreover the fact that $j \neq 1$ implies that $j_0 \neq 1$, so $e \cap V(\I_{j_0}) \neq \epsilon_{j_0}$, otherwise $e$ would be of $\A$-type "crossing". We can thus apply \Cref{prop-path1} and get an $(\epsilon_j,e \cap V(\I_{j_0}))$-extendable path $\ora{P}$ in $\A$ such that $E(\ora{P}) \subseteq \bigcup_{p=1}^M E(\I_{i_p}) \cup \bigcup_{p=2}^{M} C_{\A}(\I_{i_{p-1}},\I_{i_p})$, hence $E(\ora{P}) \subseteq E(\I)$ since $\{i_1,\ldots,i_M\}\subseteq J$ by definition of $J$. See \Cref{Type1b_Case4} (path on the right).
				
				\item Since the lowest common ancestor of $\{\I_i,i \in J_0\}$ is $\I_{i_0}$ and not $\I_j$, there exists $j'_0 \in J_0 \setminus \{j_0\}$ such that $\I_j$ is not an ancestor of $\I_{j'_0}$, so the path $\I_{i_0}=\I_{i'_1} \to \ldots \to \I_{i'_{M'}}=\I_{j'_0}$ from $\I_{i_0}$ to $\I_{j'_0}$ in $G(\A)$ satisfies $\{i_1,\ldots,i_M\} \cap \{i'_1,\ldots,i'_{M'}\}=\varnothing$ (see \Cref{Type1b_Case4} for the relative positions of the four islands in play: $\I_{i_0}$, $\I_j$, $\I_{j_0}$, $\I_{j'_0}$). As usual, we choose $j'_0$ so that this path is shortest, this way we have $e \cap V(\I_{i'_{M'}})\neq \varnothing$ and $e \cap V(\I_{i'_p})= \varnothing$ for all $1 \leq p \leq M'-1$. Since $e$ intersects several islands, we know $1 \leq |e \cap V(\I_{j'_0})| \leq k-1$. Moreover, if $j'_0 \neq 1$ then $e \cap V(\I_{j'_0}) \neq \epsilon_{j'_0}$ otherwise $e$ would be of $\A$-type "crossing". We can thus apply \Cref{prop-path1} and get an $(\epsilon_{i_0},e \cap V(\I_{j'_0}))$-extendable path $\ora{Q}$ in $\A$ such that $E(\ora{Q}) \subseteq \bigcup_{p=1}^{M'} E(\I_{i'_p}) \cup \bigcup_{p=2}^{M'} C_{\A}(\I_{i'_{p-1}},\I_{i'_p})$, hence $E(\ora{Q}) \subseteq E(\I)$ since $\{i'_1,\ldots,i'_{M'}\}\subseteq J$ by definition of $J$. See \Cref{Type1b_Case4} (path on the left).
				
				\item Let $\ora{P'}$ be the sequence obtained by reversing $\ora{P}$. Since $\ora{P}$ is an $(\epsilon_j,e \cap V(\I_{j_0}))$-extendable path, $\ora{P'}$ is an $(e \cap V(\I_{j_0}),\epsilon_j)$-extendable path. \Cref{Lemma3} applied to $A=X=\epsilon_{i_0}$, $B=e \cap V(\I_{j'_0})$, $C=e \cap V(\I_{j_0})$ and $D=\epsilon_j$, whose conditions are fulfilled since $\{i_1,\ldots,i_M\} \cap \{i'_1,\ldots,i'_{M'}\}=\varnothing$, ensures that $\ora{Q}\oplus (e) \oplus \ora{P'}$ is an $(\epsilon_{i_0},\epsilon_j)$-extendable path in $\I$ which concludes. \qedhere
			\end{itemize}
	
	\end{enumerate}
	
\end{proof}

\paragraph{$\quad$ III) $e$ is of $\A$-type "crossing"}\mbox{}\\

\indent This is the easiest case: $e$ is added as a crossing edge and the decomposition remains the same. Note that $\A \cup e$ might not be arborescent anymore (see $e=e_3$ from \Cref{Types} for example).
\begin{myproposition}\label[myproposition]{prop-type1c}
	If $e$ is of $\A$-type "crossing", then $\A \cup e$ is an archipelago with:
	\begin{itemize}[noitemsep,nolistsep]
		\item $\I(A \cup e)= \I(A)$.
		\item $\epsilon(\A \cup e)= \epsilon(\A)$.
	\end{itemize}
\end{myproposition}

\begin{proof}
	This is straightforward.
\end{proof}

\subsubsection{Formal algorithm}

\hphantom{\indent}The algorithm \textsc{Partition\_Archipelago} (\Cref{algo_main}) returns a partition of the edges that satisfies \Cref{theo-algo}. The procedures \textsc{Add\_NewCrossing}, \textsc{Add\_Other} and \textsc{Add\_Crossing} (\Cref{algo_1a,algo_1b,algo_1c}) are nothing but algorithmic translations of \Cref{prop-type1a,prop-type1b,prop-type1c} respectively. Note that islands are simply implemented as vertex sets, because their edge sets are never used.

\begin{algorithm}[H]
	\caption{\textsc{Partition\_Archipelago}$(\H,x^*)$}\label{algo_main}
	\begin{algorithmic}[1]
		\State \Initialize $V(\I_1) \gets \{x^*\}$
		\State \Define $\epsilon_1 \gets \{x^*\}$
		\State \Initialize the archipelago $\A$ with:
			\\ $\quad\,$ $E(\A) \gets \varnothing$
			\\ $\quad\,$ $\I(\A) \gets \{V(\I_1)\}$
			\\ $\quad\,$ $\epsilon(\A) \gets \{\epsilon_1\}$
			\\ $\quad\,$ $G(\A) \gets\,$ a digraph with only one vertex, labelled $\I_1$
		\State \Initialize $N \gets 1$ (index of the last created island)
		\While{there exists $e \in E(\H) \setminus E(\A)$ of $\A$-type "new crossing" or "other"}
			\If{$e$ is of $\A$-type "new crossing"}
				\State \Update $\A$ as $\A \cup e$ by performing \textsc{Add\_NewCrossing}
			\Else
				\State \Update $\A$ as $\A \cup e$ by performing \textsc{Add\_Other}
			\EndIf
		\EndWhile
		\While{there exists $e \in E(\H) \setminus E(\A)$ of $\A$-type "crossing"}
			\State \Update $\A$ as $\A \cup e$ by performing \textsc{Add\_Crossing}
		\EndWhile
		\State \Define $E_{cut} \gets \{e \in E(\H) \setminus E(\A), \text{$e$ is of $\A$-type "cut"}\}$
		\State \Define $E_{ext} \gets \{e \in E(\H) \setminus E(\A), \text{$e$ is of $\A$-type "exterior"}\}$
		\State \Return $E(\A)$, $E_{cut}$, $E_{ext}$
	\end{algorithmic}
\end{algorithm}

\begin{algorithm}[H]
	\caption{\textsc{Add\_NewCrossing}}\label{algo_1a}
	\begin{algorithmic}[1]
		\State \Define $1 \leq i_0 \leq N$ as the only index such that $e \cap V(\I_{i_0}) \neq \varnothing$
		\State \Initialize $V(\I_{N+1}) \gets e \setminus V(\I_{i_0})$
		\State \Define $\epsilon_{N+1} \gets e \setminus V(\I_{i_0})$
		\State \Update the archipelago $\A$ as follows:
			\\ $\quad\,$ $E(\A) \gets E(\A) \cup \{e\}$
			\\ $\quad\,$ $\I(\A) \gets \I(\A) \cup \{V(\I_{N+1})\}$
			\\ $\quad\,$ $\epsilon(\A \cup e) \gets \epsilon(\A) \cup \{\epsilon_{N+1}\}$
			\\ $\quad\,$ $G(\A \cup e) \gets\,$ the digraph obtained from $G(\A)$ by adding a new vertex labelled $\I_{N+1}$ and an arc $(\I_{i_0},\I_{N+1})$
		\State $N \gets N+1$
	\end{algorithmic}
\end{algorithm}

\begin{algorithm}[H]
	\caption{\textsc{Add\_Other}}\label{algo_1b}
	\begin{algorithmic}[1]
		\State \Define $J_0 \defeq \{1 \leq i \leq N \,\,\text{such that $V(\I_i) \cap e \neq \varnothing$} \}$
		\State \Define $1 \leq i_0 \leq N$ such that $\I_{i_0}=\FCA_{G(\A)}(\{\I_i,i \in J_0\})$
		\State \Define $J \defeq \bigcup_{i \in J_0} \{1 \leq j \leq N \,\,\text{such that $\I_j$ is on the path from $\I_{i_0}$ to $\I_i$ in $G(\A)$} \}$
		\State $V(\I_{i_0}) \gets \bigcup_{i \in J} V(\I_j)$
		\State \Update the archipelago $\A$ as follows:
			\\ $\quad\,$ $E(\A) \gets E(\A) \cup \{e\}$
			\\ $\quad\,$ $\I(\A) \gets \I(\A) \setminus \{V(\I_j), j\in J \setminus \{i_0\}\}$. 
			\\ $\quad\,$ $\epsilon(\A) \gets \epsilon(\A) \setminus \{\epsilon_j, j\in J \setminus \{i_0\}\}$.
			\\ $\quad\,$ $G(\A) \gets\,$ the digraph obtained from $G(\A)$ by merging the vertices $\{\I_j, j\in J\}$ into the vertex $\I_{i_0}$.
	\end{algorithmic}
\end{algorithm}

\begin{algorithm}[H]
	\caption{\textsc{Add\_Crossing}}\label{algo_1c}
	\begin{algorithmic}[1]
		\State \Update the archipelago $\A$ as follows:
			\\ $\quad\,$ $E(\A) \gets E(\A) \cup \{e\}$
	\end{algorithmic}
\end{algorithm}

\indent Let us explain the algorithm. At the start, the archipelago $\A$ consists of the empty island with entry $\{x^*\}$. We then augment $\A$ one edge at a time, by adding firstly the edges of $\A$-type "new crossing" or "other" and then the edges of $\A$-type "crossing":
\begin{itemize}[noitemsep,nolistsep]
	\item Throughout the first While loop, $\A$ is an arborescent archipelago, as guaranteed by \Cref{prop-type1a,prop-type1b}. It is very important to understand that, every time $\A$ is augmented in that loop, the vertices and entries of $\A$ may change, so the $\A$-types of the remaining edges may change as well: the $\A$-types of the edges in $E(\H)\setminus E(\A)$ must be redetermined at each iteration of that loop.
	\item Throughout the second While loop, $\A$ is an archipelago, as guaranteed by \Cref{prop-type1c}. This time, the decomposition in islands does not change during that loop (we are adding crossing edges between already existing islands) so the $\A$-types of the remaining edges do not change.
\end{itemize}
\indent That last remark proves that, after the two While loops, all remaining edges are of $\A$-type either "cut" or "exterior" (the $\A$-types "new crossing" and "other" have not reappeared during the second While loop). In conclusion, \textsc{Partition\_Archipelago} does output a partition of $E(\H)$ and is therefore correct.

\subsubsection{Time complexity}

\hphantom{\indent}Let $n=|V(\H)|$ and $m=|E(\H)|$. We now show that \textsc{Partition\_Archipelago} runs in $O(m^2 k)$ time.
\medskip
\\ \indent Let us first consider the three procedures \textsc{Add\_NewCrossing}, \textsc{Add\_Other} and \textsc{Add\_Crossing}, to figure out how much time each update of $\A$ takes. Since basic operations on data structures can be language-dependent, let us clarify: when we use a list, what matters is the ability to remove the current element in $O(1)$ time; when we use an array, what matters is the ability to access and modify any element in $O(1)$ time.
\begin{itemize}[noitemsep,nolistsep]
	\item $E(\H)\setminus E(\A)$ can be implemented as a list. Indeed, it is sensible to store $E(\H) \setminus E(\A)$ rather than $E(\A)$ since this is the set in which edges are searched for throughout. Each update consists in removing the current edge which is done in $O(1)$ time.
	\item $\I(\A)$ can be implemented as an array of size $n$ which contains, for each vertex $x \in V(\H)$, the index of the island containing $x$ (or 0 if $x \not\in V(\A)$). Each update requires going through the array once and is therefore done in $O(n)$ time.
	\item $\epsilon(\A)$ can be implemented as an array of size $n$ which contains, for each vertex $x \in V(\H)$, a 1 if $x$ is in an entry of $\A$ or a 0 otherwise. Each update requires going through the array once and is therefore done in $O(n)$ time.
	\item $G(\A)$ is an arborescence for the entire time that it is kept updated. Since $O(\frac{n}{k})$ islands are created in total (a new island can only be created during \textsc{Add\_NewCrossing}, and this requires $k-1$ previously undiscovered vertices), $G(\A)$ can be implemented as an array of size $O(\frac{n}{k})$ containing the parent of each island, i.e. for all index $i \neq 1$ it contains the only index $j$ such that $(\I_j,\I_i) \in E(G(\A))$. In \textsc{Add\_NewCrossing}, updating $G(\A)$ is clearly done in $O(1)$ time. In \textsc{Add\_Other}, updating $G(\A)$ is done in $O(n)$ time: indeed, computing $|J_0| \leq k$ paths to the root takes $O(k \times \frac{n}{k})=O(n)$ time, going through them a second time to compute $i_0$ and $J$ takes $O(k \times \frac{n}{k})=O(n)$ time again, and finally the merging process is performed in $O(\frac{n}{k})$ time since it only requires going through the array once.
\end{itemize}
\indent All in all, performing \textsc{Add\_NewCrossing}, \textsc{Add\_Other} or \textsc{Add\_Crossing} once is done in $O(n)$ time.
\medskip
\\ \indent Determining the $\A$-type of a given edge $e$ is easily done in $O(k)$ time since it boils down to determining, for all $x \in e$, which island/entry (if any) contains $x$.
\medskip
\\ \indent We can now conclude on the time complexity of \textsc{Partition\_Archipelago}:
\begin{itemize}[noitemsep,nolistsep]
	\item The initializations before the first While loop are done in $O(m+n)$ time.
	\item During the first While loop, finding an edge of $\A$-type "new crossing" or "other" and then adding it takes $O(mk+n)$ time: indeed, at most $m$ edges are gone through (with the $\A$-type being determined for each one in $O(k)$ time as we have just seen) before finally finding one of $\A$-type "new crossing" or "other" which is added in $O(n)$ time as shown above. Since at most $m$ edges of $\A$-type "new crossing" or "other" are added in total, the first While loop ends in $O(m(mk+n))=O(m^2 k + mn)$ time.
	\item During the second While loop, no $\A$-types need to be redetermined, and each update of $\A$ is done in $O(1)$ time so that this loop ends in $O(m)$ time.
	\item Finally, computing $E_{cut}$ and $E_{ext}$ at the very end of the algorithm takes $O(m)$ time.
\end{itemize}
\indent In conclusion, \textsc{Partition\_Archipelago} runs in $O(m^2 k + mn)$ time. Since the $(k-2)$-linear connected component is a subset of the connected component, it is reasonable to assume that $\H$ is connected, which implies that $m \geq \frac{n-1}{k-1}$. Therefore, we can simplify $O(m^2 k + mn)$ as $O(m^2 k)$. This ends the proof of \Cref{theo-algo}.
\medskip
\\ \indent Notice that the algorithm can easily be tweaked so as to also return a $(k-2)$-linear path from $x^*$ to $x$ for each $x \in LCC^{\,k-2}_{\H}(x^*)$. Indeed, it suffices, throughout the algorithm, to keep in memory an $(x^*,X)$-extendable path in $\A$ for each $X \subset V(\A)$ such that $1 \leq |X| \leq k-1$ and $X \not\in \epsilon(\A)$, which is possible by following the construction given in the proof of \Cref{prop-type1b}. If $k=O(1)$ then the algorithm remains in polynomial time.

\subsection{Proof of the main results}

\begin{proof}[Proof of \Cref{theo-main}]
	Let $\A,E_{cut},E_{ext}$ be as in \Cref{theo-algo}.
	\begin{itemize}
		\item Let us first show that $\A=\H[LCC^{\,k-2}_{\H}(x^*)]$. Since no edge in $E_{cut} \cup E_{ext}$ is included in $V(\A)$, we know $\A$ is an induced subhypergraph of $\H$. Moreover $V(\A) \subseteq LCC^{\,k-2}_{\H}(x^*)$ by \Cref{coro-inclusion}, so it remains to verify that $LCC^{\,k-2}_{\H}(x^*) \subseteq V(\A)$. The idea is simple: the only way to leave the archipelago is through an edge in $E_{cut}$, however a $(k-2)$-linear path in $\A$ from $x^*$ to an entry of size $k-1$ necessarily contains that entry entirely, making it impossible to then use an edge in $E_{cut}$ without violating the $(k-2)$-linearity. We now give the rigorous proof.
		\\ Suppose for a contradiction that there exists $x \in LCC^{\,k-2}_{\H}(x^*) \setminus V(\A)$. Let $\ora{P}=(e_1,\ldots,e_L)$ be a $(k-2)$-linear path from $x^*$ to $x$ in $\H$. Since $x \not\in V(\A)$, we can define $M \defeq \inf\{1 \leq p \leq L \,\,\text{such that $e_p \not\subset V(\A)$} \}$. Since all edges adjacent to $x^*$ are necessarily in $E(\A)$, we have $e_1 \subset V(\A)$ hence $M \geq 2$. Moreover $e_M$ intersects $e_{M-1} \subset V(\A)$, so $e_M \in E_{cut}$ from which $e_M \cap V(\A)=\epsilon$ for some entry $\epsilon$ of $\A$ of size $k-1$. Let $y \in e_M \cap e_{M-1} \subset \epsilon$: since $\ora{Q}\defeq (e_1,\ldots,e_{M-1})$ is a $(k-2)$-linear path from $x^*$ to $y$ in $\A$, \Cref{prop-entry1} ensures that $\epsilon \subset e_{M-1}$. Since $\epsilon \subset e_M$, this contradicts the $(k-2)$-linearity of $\ora{P}$.
		\item Any archipelago $\A'$ in $\H$ is a subhypergraph of $\H[LCC^{\,k-2}_{\H}(x^*)]=\A$, because $V(\A') \subseteq LCC^{\,k-2}_{\H}(x^*)$ by \Cref{coro-inclusion}. This shows both that $\A$ is a maximal archipelago and that it is the only one.
		\item Finally, the complexity result is obvious since computing $\H[LCC^{\,k-2}_{\H}(x^*)]=\A$ is equivalent to computing $E(\A)$. \qedhere
	\end{itemize}
\end{proof}

\section{Consequences of the algorithmic result}\label{Section4}

\subsection{Link with the Maker-Breaker positional game}

\hphantom{\indent}A \textit{positional game} is a type of combinatorial game played on a hypergraph $\H$, where two players take turns picking previously unpicked vertices of $\H$, and the winner is decided by one of several conventions. In the \textit{Maker-Breaker} convention, one player ("Maker") wins if he owns all vertices of some edge of $\H$, while the other player ("Breaker") wins if he can prevent this from happening. Note that, since both players have complementary goals, no draw is possible. The algorithmic problem consisting in deciding which player wins the Maker-Breaker game with optimal play is studied in the literature: \\

\begin{tabularx}{0.95\textwidth}{|l @{} l @{} X|}
	\hline
	\multicolumn{3}{|l|}{$\,\,\textsc{MakerBreaker}$} \\ \hline
	Input $\,$ & : & $\,$ a hypergraph $\H$. \\
	Output $\,$ & : & $\,$ YES if and only if Maker wins the Maker-Breaker game on $\H$. \\ \hline
\end{tabularx} \\

\indent The $\textsc{MakerBreaker}$ problem is trivially tractable on hypergraphs of rank 2 (Maker wins on a graph if and only if it is a matching), and is known to be PSPACE-complete on 6-uniform hypergraphs \cite{RW21}. In a separate paper \cite{GGS22}, we study the Maker-Breaker problem on hypergraphs of rank 3, in which linear paths play a crucial role. If $\H$ contains a linear path from $x$ to $y$, where Maker owns $x$ and $y$ while the other vertices of the path are free (\textit{$xy$-nunchaku}), then Maker easily wins when playing first, by forcing all of Breaker's moves along the path until Breaker is trapped. It is shown in \cite{GGS22} that Maker wins on a hypergraph of rank 3, when playing first, if and only if he has a strategy ensuring that the hypergraph contains a nunchaku at the end of one of the first four rounds of play. Therefore:

\begin{mytheorem*}\textup{\cite{GGS22}}
	$\textsc{MakerBreaker}$ on hypergraphs of rank 3 reduces polynomially to $\textsc{HypConnectivity}_{3,1}$.
\end{mytheorem*}

\indent \Cref{coro-main} thus concludes that $\textsc{MakerBreaker}$ is solvable in polynomial time on hypergraphs of rank 3. This validates a conjecture by Rahman and Watson \cite{RW20}.

\subsection{Link with {\normalfont\textsc{PAFP}}}

\subsubsection{Reducing {\normalfont\textsc{HypConnectivity}$_{k,q}$} to {\normalfont\textsc{PAFP}}}

\hphantom{\indent}A first attempt at tackling the algorithmic complexity of $\textsc{HypConnectivity}_{k,q}$, for general $1 \leq q \leq k-2$, could be the following reduction to the "Paths Avoiding Forbidden Pairs" problem known as $\textsc{PAFP}$ (sometimes $\textsc{PPFP}$ or $\textsc{PFP}$): \\

\begin{tabularx}{0.95\textwidth}{|l @{} l @{} X|}
	\hline
	\multicolumn{3}{|l|}{$\,\,\textsc{PAFP}$} \\ \hline
	Input $\,$ & : & $\,$ a bicolored graph $G$ (all edges are blue or red), and $x,y \in V(G)$. \\
	Output $\,$ & : & $\,$ YES if and only if there exists a blue induced path from $x$ to $y$ in $G$. \\ \hline
\end{tabularx} \\

\begin{mynotation}
	Let $\varphi_{k,q}$ be the function that associates to a $k$-uniform hypergraph $\H$ the bicolored graph $G$ defined by:
	\begin{itemize}[noitemsep,nolistsep]
		\item $V(G)=E(\H)$;
		\item For all distinct $e_1,e_2 \in V(G)$, there is a blue (resp. red) edge between $e_1$ and $e_2$ in $G$ if and only if $1 \leq |e_1 \cap e_2| \leq q$ (resp. if and only if $|e_1 \cap e_2| > q$).
	\end{itemize}
	Therefore $G$ is simply the line graph of $\H$ with added colors that carry information on the size of the intersections. See \Cref{Example_PAFP} for an example.
\end{mynotation}

\begin{myproposition}
	For all $k \geq 3$ and $1 \leq q \leq k-2$, $\textsc{HypConnectivity}_{k,q}$ polynomially reduces to $\textsc{PAFP}$.
\end{myproposition}

\begin{proof}
	This is clear: by definition, a sequence of edges $(e_1,\ldots,e_L)$ in $\H$ is a $q$-linear path if and only if it is a blue induced path in $\varphi_{k,q}(\H)$ ("blue" means two consecutive edges intersect on between 1 and $q$ vertices, "induced" means two non-consecutive edges do not intersect). Therefore, there exists a $q$-linear path from $x$ to $y$ in $\H$ ($x \neq y$) if and only if there exist edges $e_x \ni x$ and $e_y \ni y$ in $\H$ such that there exists a blue induced path between $e_x$ and $e_y$ in $\varphi_{k,q}(\H)$.
\end{proof}

\indent However, $\textsc{PAFP}$ is known to be NP-complete in general \cite{GMO76}. In fact, unless P=NP, there is no linear approximation ratio for the minimum number of red edges induced by a blue path between two given vertices \cite{HKK12}. For the problem on directed graphs (the blue edges are directed arcs), which is by far the most studied version in the literature, a few tractable cases are known but they are of little help to us:
\begin{itemize}[noitemsep,nolistsep]
	\item It is shown in \cite{Yin97} that the problem is tractable if the red edges form a matching and a \textit{skew symmetry condition} is satisfied. Even though the undirected version is also true with basically the same proof, it does not solve $\textsc{HypConnectivity}_{k,q}$ since a general bicolored graph in $\Im(\varphi_{k,q})$ does not satisfy these conditions (nor does it easily reduce to one that does).
	\item Other tractable cases are addressed in \cite{CKT01} and \cite{KP09}, however they are very specific to directed acyclic graphs.
\end{itemize}

\begin{figure}[htbp]
	\centering
	\includegraphics[scale=.6]{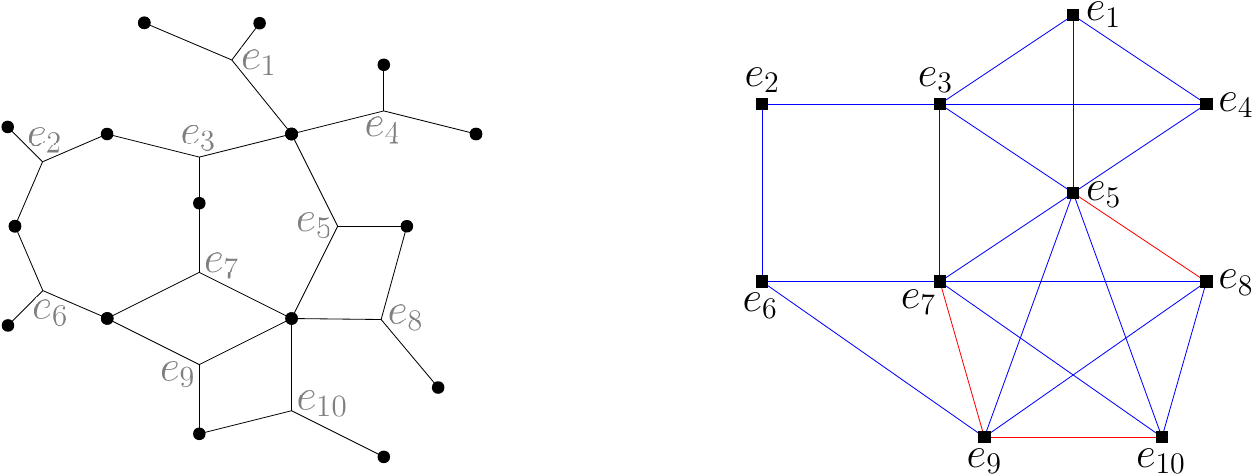}
	\caption{On the left: a 3-uniform hypergraph $\H$. On the right: the bicolored graph $G=\varphi_{3,1}(\H)$.}\label{Example_PAFP}
\end{figure}

\subsubsection{Reducing some instances of {\normalfont\textsc{PAFP}} to {\normalfont\textsc{HypConnectivity}$_{k,q}$}}

\hphantom{\indent}Instead, now that we know $\textsc{HypConnectivity}_{k,k-2}$ is solvable in polynomial time for all $k \geq 3$, it is interesting to turn the tables and examine the implications on $\textsc{PAFP}$:

\begin{mytheorem}
	$\textsc{PAFP}$ is tractable on bicolored graphs in $\bigcup_{k \geq 3}\Im(\varphi_{k,k-2})$ for which a preimage can be computed in polynomial time.
\end{mytheorem}

\begin{proof}
	Let $G=\varphi_{k,k-2}(\H)$ for some $k$-uniform hypergraph $\H$, and let $e,e' \in V(G)=E(\H)$ be distinct. As we have seen before, the blue induced paths between $e$ and $e'$ in $G$ are exactly the $(k-2)$-linear paths $(e=e_1,\ldots,e_L=e')$ in $\H$. Since $\textsc{HypConnectivity}_{k,k-2}$ requires a start vertex and an end vertex in its input, define, for all $x \in e$ and $y \in e'$, the hypergraph $\H_{x,y}$ obtained from $\H$ by removing all edges adjacent to $x$ and $y$ other than $e$ and $e'$, so that any $(k-2)$-linear path from $x$ to $y$ in $\H_{x,y}$ necessarily starts with $e$ and ends with $e'$. There exists a blue induced path between $e$ and $e'$ in $G$ if and only if there exist $x \in e$ and $y \in e'$ such that there is a $(k-2)$-linear path from $x$ to $y$ in $\H_{x,y}$, which concludes since $\textsc{HypConnectivity}_{k,k-2}$ is solvable in polynomial time.
\end{proof}

\indent Therefore, any sufficient condition for a bicolored graph $G$ to be in $\Im(\varphi_{k,k-2})$ for some $k \geq 3$, if it can be checked in polynomial time and comes with a way to reconstruct a preimage hypergraph in polynomial time, would add to the very short list of known tractable cases for $\textsc{PAFP}$.
\medskip
\\ \indent For standard (i.e. non-colored) line graphs, the recognition problem has been studied extensively. Line graphs of graphs are characterized by a finite list of forbidden induced subgraphs ("FIS") \cite{Bei70}. Line graphs of hypergraphs, on the other hand, are notoriously difficult to recognize. There is no finite FIS characterization for line graphs of $k$-uniform hypergraphs if $k \geq 3$ \cite{Lov77}, and this recognition problem is even known to be NP-complete for $k=3$ \cite{PRT81}. However, adding information about the size of the pairwise intersections of (hyper)edges, instead of simply telling which ones are non-empty, changes the problem. For example, if all these sizes are given and in $\{0,1\}$ (which is equivalent to asking the hypergraph to be \textit{linear}) then, while remaining NP-complete for $k=3$ \cite{PRT81} \cite{HK97}, the problem becomes easier in some cases:
\begin{itemize}[noitemsep,nolistsep]
	\item For $k=3$, there is a finite FIS characterization for line graphs of 3-uniform linear hypergraphs if the minimum vertex-degree of the graph is at least 69, as well as a polynomial time algorithm to reconstruct the hypergraph in the positive case \cite{NRS82}. This bound has since been improved from 69 to 16 for the finite FIS characterization and 10 for the tractability of the recognition problem \cite{SST09}. There is no analogous result for $k \geq 4$, no matter what constant lower bound is put on the minimum vertex-degree \cite{MT97}.
	\item For any $k \geq 3$, there is a finite FIS characterization for line graphs of $k$-uniform linear hypergraphs if the minimum edge-degree of the graph is at least $f(k)$, where $f$ is a polynomial function, as well as a polynomial (whose power increases with $k$) time algorithm to reconstruct the hypergraph in the positive case \cite{NRS82}. This result has been generalized by replacing the linearity of the hypergraph by any constant upper bound on its \textit{multiplicity} \cite{BGM21}.
\end{itemize}
\hphantom{\indent}These results bring some hope of a finite FIS characterization for bicolored line graphs under some similar restriction over the minimum vertex-degree or edge-degree of the graph, and of a way to reconstruct a preimage in polynomial time which we crucially need. The case $k=3$ is the most promising because the exact size of each intersection is also given (in $\{0,1,2\}$: 0 = no edge, 1 = blue edge, 2 = red edge), although it is NP-complete in general since instances with all blue edges correspond to the 3-uniform linear case for standard line graphs which we know is NP-complete. \Cref{Forbidden} features some induced bicolored subgraphs that cannot appear in a bicolored graph from $\bigcup_{k \geq 3}\Im(\varphi_{k,k-2})$. For instance, an induced red path on three vertices is impossible because, in a $k$-uniform hypergraph with $k \geq 3$, if $|e_1 \cap e_2|=|e_2 \cap e_3|=k-1$ then $|e_1 \cap e_3| \geq k-2>0$.

\begin{figure}[htbp]
	\centering
	\includegraphics[scale=.5]{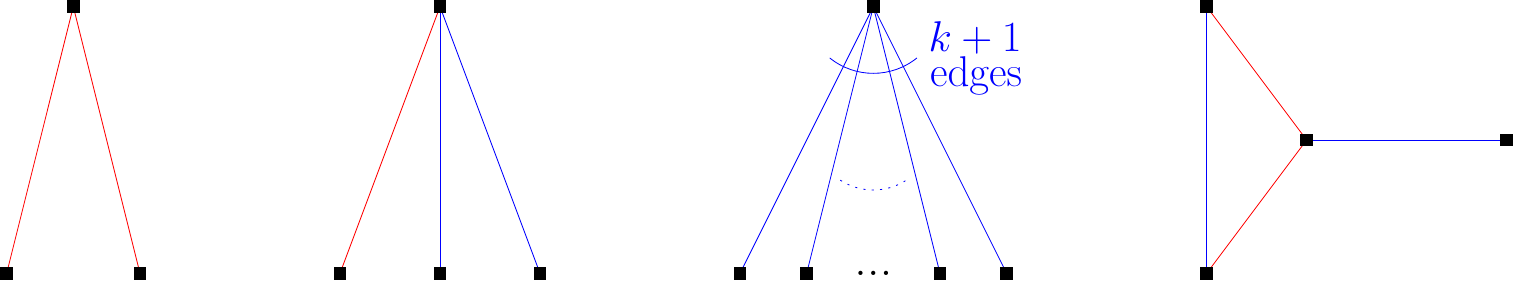}
	\caption{Some induced subgraphs that cannot appear in $G \in \Im(\varphi_{k,k-2})$.}\label{Forbidden}
\end{figure}

\section*{Conclusion and perspectives}

\hphantom{\indent}In this paper, we have introduced $q$-linear paths in hypergraphs of rank $k$, and in the case $q=k-2$ we have described the structure of the $(k-2)$-linear connected components as well as a polynomial time algorithm to compute them. The time complexity in $O(m^2k)$ might be optimal, since it seems difficult to avoid an "accept or put aside" process on the edges where each edge is potentially examined $O(m)$ times, and the mere computation of the intersection of two edges is in $O(k)$ time.
\medskip
\\ \indent What about other values of $q$? The linear case $q=1$ is of particular interest, since linear paths appear in numerous other problems. However, if we want to try and generalize our techniques while maintaining a time complexity that is polynomial in $k$, it might be more reasonable to look at the case $q=k-c$ where $c \geq 3$ is a constant, with adapted definitions of islands and archipelagos (whose entries would be of size between $k-c+1$ and $k-1$). As an illustration of the difficulties that can be encountered during the algorithm, consider the case $k=4$ and $q=1$, where at some point an edge $e=\{x,y,z,t\}$ is discovered with $x,y$ already known vertices from different islands and $z,t$ unknown: on one hand $e$ could be part of a new merged island (since $x,y \in e$), but on the other hand $e$ could be a crossing edge towards a new island with entry $\{z,t\}$ (since $z$ and $t$ are not separated), and it seems hard to conciliate the two.
\medskip
\\ \indent The bicolored line graph recognition problem is open. As mentioned in \Cref{Section4}, the added information on the size of the pairwise intersections of edges might make this problem somewhat easier compared to standard line graphs, especially in the case $k=3$. The characterization of line graphs of hypergraphs by a Krausz partition into cliques \cite{NRS82} is easily adaptable to the bicolored version. Some characterizations by finite families of induced subgraphs from \cite{NRS82} and their proofs might be adaptable as well, which would yield new classes of tractable instances for \textsc{PAFP}. Looking beyond applications to \textsc{PAFP}, a general \textit{weighted line graph} recognition problem, where each edge of the graph would wear a number between 1 and $k-1$ indicating the exact size of the corresponding intersection, seems interesting in itself.

\end{document}